\documentclass[onecollarge,smallcondensed,smallextended]{svjour3}

\smartqed  

\usepackage{booktabs,array,multirow}
\usepackage{amssymb}
\usepackage[intlimits]{amsmath}
\usepackage{rotating}
\usepackage{amsfonts}
\usepackage{dsfont}
\usepackage{graphicx}
\usepackage{newlfont}
\usepackage{mathrsfs}
\usepackage{makeidx}
\usepackage[utf8]{inputenc}
\usepackage{rotating}
\usepackage{tabularx}
\usepackage{txfonts}
\usepackage{mathptmx}
\usepackage{float}
\newenvironment{notation}{\par\noindent\textbf{Notation:} }{\vspace{1ex}\par}

\numberwithin{equation}{section}

 \newcommand{\x}{\mathrm{x}}
 \newcommand{\y}{\mathrm{y}}

 \newcommand{\E}{\mathbb{E}}

 \newcommand{\N}{\mathbb{N}}

 \newcommand{\Real}{\mathbb{R}}

 \newcommand{\norm}[1]{\left\Vert#1\right\Vert}
 \newcommand{\warunk}[1]{{\hat #1}}
 \newcommand{\przybl}[1]{{#1^\ast}}

\journalname{journal}

\begin{document}

\title{Approximations of Bond and Swaption Prices in a~Black-Karasi\'{n}ski Model}
\titlerunning{Approximations of Bond and Swaption Prices in a~Black-Karasi\'{n}ski Model}        
\author{A. Daniluk, R. Muchorski}
\institute{A. Daniluk, Jagiellonian Univeristy, \email{andrzej.daniluk@gmail.com} \\
R. Muchorski, TUiR Allianz Polska, \email{rafal.muchorski@gmail.com}}
\date{\textbf{May 31, 2015}}
\maketitle

\begin{abstract}
We derive semi-analytic approximation formulae for bond and swaption prices in a Black-Karasi\'{n}ski interest rate model. Approximations are obtained using a novel technique based on the Karhunen-Lo\`{e}ve expansion. Formulas are easily computable and prove to be very accurate in numerical tests. This makes them useful for numerically efficient calibration of the model. 

\keywords{Black-Karasi\'{n}ski Model \and Karhunen-Lo\`{e}ve expansion}
\subclass{91G30 \and 60H30 \and 41A99}
\end{abstract}

\section*{Acknowledgement}
Preprint of an article submitted for consideration in International Journal of Theoretical and Applied Finance \copyright\ 2015, copyright World Scientific Publishing Company, URL: http://www.worldscientific.com/worldscinet/ijtaf

\section{Introduction}\label{intro}

Short-rate models are of fundamental importance in the quantitative field of finance, as they provide a comprehensive mathematical framework for pricing interest rate or credit derivatives \cite{piterbarg},\cite{okane}. The diversity of model structures and assumptions, enable us to choose the most appropriate approach when dealing with specific pricing issues. The basic Gaussian affine models, such as Vasicek \cite{vasicek} and Hull and White \cite{hull-white} gained interest among practitioners due to their analytical tractability and transparency, with closed-form pricing formulas available on hand. There is however, a trade-off between such advantages and implausible model forecasts, which allow negative interest rates. Some others, such as the Cox, Ingresoll and Ross model \cite{cir}, despite having the property of positive rates, often provide unrealistic outcomes by imposing a lower positive bound for the par swap rate \cite{piterbarg}.
Non-negativity of interest rates does not seem so important, or may even be undesirable in today's low-interest rate environment. However, in the context of default intensity modeling, negative hazard rates are generally not feasible due to lack of consistency with arbitrage-free assumptions.

A model developed by F. Black and P. Karasi\'{n}ski (BK) in 1991, also known as the "exponential Vasicek model" \cite{bk}, overcomes the problem of negative rates. It postulates log-normality of short rates, motivated by the fact that the market standard Black formulae for caps and swaptions are based on log-normal distributions of relevant rates. Moreover, it possesses rather good fit-to-data properties, especially concerning the swaption volatility surface. Unfortunately, in this model, exact analytical formulae for swaptions, or even for zero-coupon bond prices, do not exist. This lack of analytical tractability requires the use of computationally intensive and time-absorbing numerical methods (PDE or Monte Carlo). This virtually precludes efficient model calibration and seriously narrows areas of potential model applicationss.

In response to challenges related to implementation of the BK model, there have been several attempts to obtain reliable analytical approximations of zero-coupon bond or swaption prices. In particular, Tourruc\^{o}o {\it et al.} \cite{tourrucoo} proposed approximate formulas for zero-coupon bonds in a one-factor model, derived in the limit of small volatility by applying the regular asymptotic expansion of a transformed bond PDE \cite{piterbarg}. Antonov
and Spector \cite{antonov} went further and came up with a generalised multi-factor BK model. By performing a regular asymptotic expansion of the PDE, they provide approximations for both zero-coupon bond and European swaption prices. Nevertheless, both approaches consider small volatility cases, which are not plausible in many financial applications, in particular in credit markets, where lognormal volatilities of default intensity as high as 100\% are frequently observed \cite{okane}.
A different approach towards deriving approximations has been adopted by Stehlikova \cite{stehlikova}, who developed small time expansions for one factor models. Prices of zero-coupon bonds were represented by means of a Taylor series expansion with coefficients represented in a closed form, obtained via specific recurrent relations. Another approximation concept, which originated from chemical physics under the name of "the exponent expansion", was introduced to Finance by Capriotti \cite{capriotti1} and applied to the calculation of transition probabilities and Arrow-Debreu prices, for several diffusion processes. This approach appeared to provide very accurate approximations and was further pursued by Stehlikova, Capriotti \cite{capriotti2} in the context of the BK model. On the basis of the exponent expansion, the authors proposed representing the price of a zero-coupon bond in the form of a power series in time, that can easily be computed by means of a recursion involving only simple one-dimensional integrals. For larger time horizons, the exponent expansion can be combined with a fast numerical convolution to obtain more accurate results. However, these approximation formulas are quite complex and there is generally no guidance on the selection of the proper truncation order of the exponent expansion, to obtain a sufficient degree of approximation accuracy.

In this paper, we propose a totally different approach towards approximating zero-coupon bond and swaption prices. The concept is based on a novel technique, applying the Karhunen-Lo\`{e}ve representation
\cite{loeve} of the Ornstein-Uhlenbeck process and one other related process, that appear in the context of a BK model. Initially, we provide semi-analytic approximations of zero-coupon bond prices. Afterwards, we derive analogous approximations for swaptions using a conceptually similar, yet more elaborate approach. The formulae are easy to implement, computationally fast and provide very accurate approximations for the vast majority of parameter settings. As zero-coupon bonds and swaptions (caps/floors as a special case) are basic instruments used for calibration of the model, our results can be used directly for that purpose, substantially improving the speed of calibration.

The structure of this paper is organized as follows. In the second section we introduce the mathematical framework established for deriving price approximations in a BK model. The third section provides two kinds of approximations of a zero-coupon bond, with all related derivations included. The fourth section is composed in a similar manner and deals with swaption pricing. The final fifth section contains numerical results of approximations for bonds and swaptions for several selected sets of model parameters, as well as comparisons with alternative approximations for bond prices based on papers 
\cite{tourrucoo},\cite{capriotti2}. For the readers convenience, most technical and purely mathematical considerations are gathered in the Appendix.

\section{Mathematical preliminaries}

We denote by $(r_t)_{t\geq 0}$ the process of the short rate and by $l_t\coloneqq\ln(r_t)$ it's natural logarithm. We assume that the process $(l_t)_{t\geq 0}$ follows dynamics postulated in a Black-Karasi\'{n}ski model \cite{bk} i.e.
\begin{equation}\label{eq:123}
d l_t=\big(a(t)-b\,l_t\big)\,dt+\sigma\,dW_t,
\end{equation}
where $\sigma, b$ are positive constants, $a(t)$ is some deterministic function of time and $(W_t)_{t\geq 0}$ is a Wiener process under the spot measure. Let us introduce some notations and recall a few facts.\\

\begin{notation}
For $u\geq 0$, $v\geq 0$ we set\\
\begin{equation}	
A(u,u+v) \coloneqq \int_u^{u+v}e^{-b(u+v-s)}a(s)ds,
\end{equation}
\begin{equation}
\bar{r}_{u,v} \coloneqq r_u^{e^{-bv}}\exp(A(u,u+v)),
\end{equation}
\begin{equation}
\quad X^{(u)}_t \coloneqq \frac{1}{\sigma}\big(l_{u+t}-e^{-bt}l_u-A(u,u+t)\big).
\end{equation}
In the specific case $u=0$ we denote $X_t \coloneqq X^{(0)}_t$.
\end{notation}

\begin{proposition}\label{prop1}
For any $u,t\geq 0$
\begin{equation}
r_{u+t}=\bar{r}_{u,t}\exp\big(\sigma X^{(u)}_t\big),
\end{equation}
where $\big(X^{(u)}_t\big)_{t\geq 0}$ is an Ornstein-Uhlenbeck process satisfying
\begin{equation}\label{eq:58}
X^{(u)}_0=0, \quad dX^{(u)}_t=-bX^{(u)}_t\,dt+dW^{(u)}_t,
\end{equation}
where $W^{(u)}_t \coloneqq W_{u+t}-W_u$ is a Wiener process.
\end{proposition}
\begin{proof}
Using \eqref{eq:123} and elementary calculus we check that $\big(X^{(u)}_t\big)_{t\geq 0}$ is an Ornstein-Uhlenbeck process satisfying \eqref{eq:58}. By virtue of introduced notations we have 
\begin{equation}
r_{u+t}=\exp\left(A(u,u+t)+e^{-bt}\ln(r_u)+\sigma X^{(u)}_t\right)=\bar{r}_{u,t}\exp\big(\sigma X^{(u)}_t\big).
\end{equation}
\end{proof}\qed

\begin{remark}
Processes $\big(W^{(u)}_t\big)_{t\geq 0}$ and $\big(X^{(u)}_t\big)_{t\geq 0}$ are independent of $\mathcal{F}_u$. 
\end{remark}

\begin{remark}
Quantity $\bar{r}_{u,v}$ can be interpreted as the value of the short term rate $r_{u+v}$ in the absence of volatility. 
It is also the dominant of the distribution of $r_{u+v}$ under $\mathcal{F}_u$. 
\end{remark}

\begin{remark}
Ornstein-Uhlenbeck process $(X_t)_{t\geq 0}$  is a centred Gaussian process with a covariance function \cite{karatzas-shreve}
\begin{equation}\label{eq:59}
K(s,t) \coloneqq \frac{1}{2b}e^{-b|t-s|}-\frac{1}{2b}e^{-b(t+s)}
\end{equation}
and a variance function
\begin{equation}
V(s) \coloneqq K(s,s)=\frac{1}{2b}\left(1-e^{-2bs}\right)
\end{equation}
\end{remark}
  
For further purpose we introduce yet another process

\begin{notation}
Let us define an Ornstein-Uhlenbeck bridge process $\big(\warunk{X}_t\big)_{t\in[0,T]}$ as
\begin{equation}\label{eq:102}
\warunk{X}_s \coloneqq X_s-\frac{K(s,T)}{V(T)}X_T, \quad s\in[0,T].
\end{equation}
\end{notation}

\begin{remark}\label{eq:247}
Ornstein-Uhlenbeck bridge $\big(\warunk{X}_t\big)_{t\in[0,T]}$ is a centred Gaussian process satisfying $\warunk{X}_0=\warunk{X}_T=0$, with a covariance function 
\begin{equation}\label{eq:91}
\warunk{K}(s,t) \coloneqq K(s,t)-\frac{K(s,T)K(t,T)}{V(T)}
\end{equation}
and a variance function
\begin{equation}
\warunk{V}(s) \coloneqq \warunk{K}(s,s)=V(s)\Bigg(1-\frac{e^{-2b(T-s)}}{1-e^{-2bT}}\Bigg).
\end{equation}

Moreover for any $s\in[0,T]$ a random vector $(X_T,\warunk{X}_s)$ has a joint Gaussian distribution and by definition $\E(X_T\warunk{X}_s)=0$, hence $\warunk{X}_s$ is independent of $X_T$.
\end{remark}

Now we can state main theorems, on which our approximations are based on. Here we present their assertions only. For detailed proofs please refer to the Appendix.

\begin{theorem}[Karhunen-Lo\`{e}ve Theorem, Lo\`{e}ve \cite{loeve}]\label{Karh}
Let $(X_t)_{t\in[a,b]}$ be a centred stochastic process with a covariance function $K$. Then $X_t$ admits the
expansion (called the Karhunen-Lo\`{e}ve representation or Karhunen-Lo\`{e}ve expansion)
\begin{equation}\label{eq:46}
X_t=\sum_{n=0}^\infty\sqrt{\lambda_n}f_n(t)Z_n\quad\textrm{a.s.},
\end{equation}
where
\begin{equation}\label{eq:47}
Z_n=\frac{1}{\sqrt{\lambda_n}}\int_a^bX_tf_n(t)\,dt,\quad n\geq 0
\end{equation}
are orthogonal random variables such that $\E Z_n=0$ and $\E Z_n^2=1$. Moreover $\{f_n(t)\}_{n\geq 0}$ form an orthonormal basis in
$L^2([a,b])$ consisting of the eigenfunctions (corresponding to
non-zero eigenvalues) of a Fredholm operator $\mathcal{T}$, i.e.
\begin{equation}\label{eq:69}
\mathcal{T}f_n=\lambda_nf_n, \quad n\geq 0.
\end{equation}
The series \eqref{eq:46} converges a.s. and uniformly for $t\in[a,b]$ in the norm
$\|Y(t)\| \coloneqq \left(\E Y^2(t)\right)^{1/2}$. In particular, if $(X_t)_{t\in[a,b]}$ is a centred Gaussian process, then
$Z_n$'s given by \eqref{eq:47} are independent $N(0,1)$ random variables.
\end{theorem}

The following two theorems give the explicit form of Karhunen-Lo\`{e}ve expansion for the Ornstein-Uhlenbeck process and Ornstein-Uhlenbeck bridge.

\begin{theorem}\label{thm1}
Let $(X_t)_{t\in[0,\tau]}$ be an Ornstein-Uhlenbeck process in the interval $[0,\tau]$ satisfying the equation \eqref{eq:58}. Then its Karhunen-Lo\`{e}ve expansion is of the form
\begin{equation}\label{eq:61}
X_t=\sum_{n=0}^\infty \sqrt{\lambda_{n}(\tau)}f_{n,\tau}(t)Z_n,
\end{equation}
where $Z_n$ are independent, identically distributed normal $N(0,1)$ random variables and $f_{n,\tau},\lambda_{n}(\tau)$ are given by
\begin{equation}
f_{n,\tau}(t)=\sqrt{\frac{2}{\tau + b\lambda_{n}(\tau)}}\sin(\omega_{n}(\tau)t),\quad \lambda_{n}(\tau)=\frac{1}{b^2 + \omega_{n}(\tau)^2},
\end{equation}
where $\omega_{n}(\tau)$ is the unique solution of the equation
\begin{equation}
\omega\cot(\omega \tau)=-b
\end{equation}
in the interval 
$\bigg(\Big(n+\frac{1}{2}\Big)\frac{\pi}{\tau},(n + 1)\frac{\pi}{\tau}\bigg),\ n\in\N\cup\{0\}.$
\end{theorem}

\begin{theorem}\label{thm2}
Let $(\warunk{X}_t)_{t\in[0,T]}$ be an Ornstein-Uhlenbeck bridge process in the interval $[0,T]$ satisfying \eqref{eq:102}. Then its Karhunen-Lo\`{e}ve expansion is of the form
\begin{equation}\label{eq:93}
\warunk{X}_t=\sum_{n=1}^\infty \sqrt{\warunk{\lambda}_{n}(T)}\warunk{f}_{n,T}(t)\warunk{Z}_n,
\end{equation}
where $\warunk{Z}_n$ are independent, identically distributed normal $N(0,1)$ random variables and $f_{n,T},\warunk{\lambda}_{n}(T)$ are given by

\begin{equation}	
\warunk{f}_{n,T}(t)=\sqrt{\frac{2}{T}}\sin\bigg(\frac{n\pi t}{T}\bigg),\quad \warunk{\lambda}_n(T)=\frac{T^2}{b^2T^2+n^2\pi^2}.
\end{equation}
\end{theorem}

For the purpose of further elaborations we also introduce a few important notations and recall some useful theorems and lemmas.\\

\begin{notation}
Let $H_n$ denote the probabilistic Hermite polynomial of degree $n$ and let $h_{n,1},\ldots,h_{n,n}$ denote its zeros (in ascending order).
\end{notation}

\begin{notation}
For a function $f:\Real \rightarrow \Real$ let $\mathfrak{L}_n(f)$ denote its Lagrange interpolation polynomial with nodes $h_{n,k}$, i.e.
\begin{displaymath}
\left[\mathfrak{L}_n(f)\right](z) \coloneqq \sum_{k=1}^{n} f(h_{n,k}) \frac{H_n(z)}{(z-h_{n,k})\prod_{j\neq k}(h_{n,k}-h_{n,j})}.
\end{displaymath}
\end{notation}

\begin{notation}
For a function $f:\Real \rightarrow \Real$ let $\mathfrak{Q}_n(f)$ denote the Gauss-Hermite quadrature of degree $n$, rescaled to the probabilistic convention, with abscisas $h_{n,k}$ and weights $w_{n,k}$, i.e. 
\begin{equation}
\mathfrak{Q}_n(f) \coloneqq \sum_{k=1}^{n} w_{n,k}f(h_{n,k}),
\end{equation}
where 
\begin{equation}
w_{n,k}=\frac{2^{n-1}n!}{n^2[H_{n-1}(h_{n,k})]^2}.
\end{equation}
\end{notation}

\begin{remark}
It is well known that for functions $f \in \mathbb{L}^2\left(\Real,\frac{1}{\sqrt{2\pi}}e^{-\frac{1}{2}x^2}dx\right)$ asymptotically holds 
\begin{equation}\label{eq:233}
\lim_{n \to +\infty}\mathfrak{Q}_n(f)=\frac{1}{\sqrt{2\pi}}\int_{-\infty}^{+\infty} f(z) e^{-\frac{1}{2}z^2} dz,
\end{equation}
which becomes an equation for polynomials of degree up to $2n-1$.
\end{remark}

\begin{notation}
Define the operator $\mathfrak{H}: \Real[X] \rightarrow \Real[X]$ and the functional $\mathfrak{h}: \Real[X] \rightarrow \Real$ such that for polynomial $W(z) = \sum_{k=0}^n w_k z^k$
\begin{equation}
\left[\mathfrak{H}(W)\right](z) \coloneqq \sum_{k=0}^{n-1} u_k z^k,\quad \mathfrak{h}(W) \coloneqq u_{-1},
\end{equation}
where coefficients $u_k$ are defined recursively as follows
\begin{equation}
u_{n} =  0, \quad u_{n+1} =  0,
\end{equation}
\begin{equation}
u_k = (k+2) u_{k+2} + w_{k+1}, \quad k=n-1,\ldots,-1.
\end{equation}
\end{notation}

\begin{remark}
Note that $\mathfrak{H}, \mathfrak{h}$ are properly defined without the implicit assumption $deg(W)=n$. Indeed, for any $m>n$ setting $0=w_m=...=w_{n+1}$ and $\przybl{W}(z)= \sum_{k=0}^m w_k z^k$ we obtain
\begin{equation}
\mathfrak{H}(\przybl{W})=\mathfrak{H}(W), \quad \mathfrak{h}(\przybl{W})=\mathfrak{h}(W).
\end{equation}
\end{remark}

\begin{lemma}\label{lem55}
For any polynomial $W(z)$
\begin{equation}
-\int W(z) e^{-\frac{1}{2}z^2} dz = \left[\mathfrak{H}(W)\right](z) e^{-\frac{1}{2}z^2} - \mathfrak{h}(W)\Phi(z)+C.
\end{equation}
\end{lemma}
For the proof refer to Appendix.\\

\begin{lemma}\label{lem666}
For any continuous, bounded function $f:\Real \rightarrow \Real$
\begin{equation}
\lim_{n \to +\infty} \mathfrak{L}_n(f) = f,
\end{equation}
where convergence holds in the $\mathbb{L}^2\left(\Real,\frac{1}{\sqrt{2\pi}}e^{-\frac{1}{2}x^2}dx\right)$ norm.
\end{lemma}
For the proof refer to the Appendix.

\section{Approximation for zero-coupon bonds}

In this section we present two semi-analytical formulae for approximate pricing of zero-coupon bonds in a BK model. The derivation of the first formula is based on general pricing fundamentals and application of Theorem \ref{thm1}. The second formula adds some additional approximations and simplifications, making it more usable without much loss of accuracy.
\\
\begin{notation}
Let $T, \tau > 0, n \geq 0$. We define the following functions
\begin{equation}
F_{n,\tau}(t,z_0,\ldots,z_n) \coloneqq \exp \left( \sigma \sum_{k=0}^n \sqrt{\lambda_{k}(\tau)}f_{k,\tau}(t)z_k \right),
\end{equation}
\begin{equation}
G_{n,\tau}(t) \coloneqq \exp\left(\frac{\sigma^2}{2}\Big(V(t)- \sum_{k=0}^n \lambda_{k}(\tau)f_{k,\tau}(t)^2\Big)\right),
\end{equation}
\begin{equation}
 I_{n,\tau}(z_0,\ldots,z_n) \coloneqq \int_0^\tau \bar{r}_{T,t} G_{n,\tau}(t) F_{n,\tau}(t,z_0,\ldots,z_n)\,dt.
\end{equation}
\end{notation}

\newtheorem{aprox}{Approximation}
\begin{aprox}\label{ap7}
The sequence 
\begin{equation}
B_n(T,T+\tau) \coloneqq \frac{1}{(\sqrt{2\pi})^{n+1}}\int_{-\infty}^{+\infty} \ldots \int_{-\infty}^{+\infty} \exp \left( -I_{n,\tau}(z_0,\ldots,z_n) -\frac{z_0^2 + \ldots + z_n^2}{2} \right) \,dz_0 \ldots dz_n
\end{equation}
provides approximation
\begin{equation}
B(T,T+\tau) = \lim_{n \rightarrow +\infty} B_n(T,T+\tau).
\end{equation}
\end{aprox}

\begin{proof}
Considering the process $\Big(X^{(T)}_t\Big)_{t\in[0,\tau]}$ as defined in Proposition \ref{prop1} and it's Karhunen-Lo\`{e}ve expansion from Theorem \ref{thm1}, we have
\begin{equation}
r_{T+t} = \bar{r}_{T,t}\exp\Big(\sigma X^{(T)}_t\Big) = \przybl{r_{n,T+t}} R_{n,\tau}(t),
\end{equation}
where
\begin{equation}
\przybl{r_{n,T+t}} = \bar{r}_{T,t} G_{n,\tau}(t) F_{n,\tau}(t,Z_0,\ldots,Z_n),
\end{equation}
\begin{equation}
R_{n,\tau}(t) = G_{n,\tau}^{-1}(t) \exp \left( \sum_{k=n+1}^{\infty} \sigma  \sqrt{\lambda_{k}(\tau)}f_{k,\tau}(t)Z_k \right).
\end{equation}
Notice that $\bar{r}_{T,t}$ is a deterministic, continuous function of $t$, $\przybl{r_{n,T+t}}$ is $\mathcal{F}_T$-independent and 
\begin{equation}\label{eq:33345}
\E\przybl{r_{n,T+t}}=\bar{r}_{T,t}\exp\left(\frac{\sigma^2}{2}V(t)\right), \quad t \in [0,\tau], n\geq 0,
\end{equation}
\begin{equation}\label{eq:77}
\E R_{n,\tau}(t)=1, \quad t \in [0,\tau], n\geq 0,
\end{equation}
\begin{equation}\label{eq:11122}
\E R_{n,\tau}^2(t)=\exp\left(\sigma^2\sum_{k=n+1}^{+\infty}\lambda_{k}(\tau)f_{k,\tau}(t)^2\right), \quad t \in [0,\tau], n\geq 0.
\end{equation}
Fix $t \in [0,\tau]$. Recalling formulas from Theorem \ref{thm1}, we notice that 
\begin{equation}
\sup_{t\in[0,\tau]}f_{k,\tau}(t)^2 \leq \frac{2}{\tau}, \qquad
\lambda_{k}(\tau) \leq \frac{4\tau^2}{(2k+1)^2\pi^2},
\end{equation}
therefore from \eqref{eq:11122}
\begin{equation}\label{eq:2333}
\E R_{n,\tau}^2(t)\leq \exp\left(\frac{8\tau\sigma^2}{\pi^2}\sum_{k=n+1}^{+\infty}\frac{1}{(2k+1)^2}\right) = \colon M_n.
\end{equation}
Hence by Jensen's inequality 
\begin{equation}\label{eq:4321}
\E\left|R_{n,\tau}(t)-1\right| \le \sqrt{\E\left(R_{n,\tau}(t)-1\right)^2} = \sqrt{\E R_{n,\tau}^2(t)-1} \le \sqrt{M_n^2 - 1}.
\end{equation}
Finally, since
\begin{equation}
B_n(T,T+\tau) = \E \exp \left(- \int_0^\tau \przybl{r_{n,T+t}}\,dt \Bigg| \mathcal{F}_T\right),
\end{equation}
then applying the inequality $|\exp(-x)-\exp(-y)| \leq |x-y|$ for $x,y \geq 0$, switching the order of integration and using \eqref{eq:33345}, \eqref{eq:4321} we obtain
\begin{equation}
\begin{split}
&\big|B(T,T+\tau)-B_n(T,T+\tau)\big| \leq\E\left(\Bigg|\int_{0}^{\tau}\left(r_{T+t}-\przybl{r_{n,T+t}}\right)dt\Bigg|\Bigg|\mathcal{F}_T\right)\\
&\leq\E\left(\int_{0}^{\tau}\left|r_{T+t}-\przybl{r_{n,T+t}}\right|dt\Bigg|\mathcal{F}_T\right)
=\int_{0}^{\tau}\E\przybl{r_{n,T+t}}\E|R_{n,\tau}(t)-1|dt \\
&\leq \tau \sup_{t \in [0,\tau]}\bar{r}_{T,t} \exp\left(\frac{\sigma^2}{2}V(T)\right)
\sqrt{M_n^2 - 1}.
\end{split}
\end{equation}
and since $M_n \to 0$, we get our assertion.
\end{proof}\qed

\begin{remark}
Although the formula from Approximation \ref{ap7} may seem quite complicated, it is applicable and for small $n$ can easily be calculated numerically. In particular, the external integrals can be evaluated with high accuracy by the use of the (probabilistic) Gauss-Hermite quadrature, whereas the internal integral $I_{n,\tau}$ can be calculated using the Romberg method or Legendre quadrature (provided that $a(t)$ satisfies sufficient smoothness conditions).
\end{remark}

\begin{remark}
Approximation \ref{ap7} presents the convergence property of the sequence \\ $\left\{B_n(T,T+\tau)\right\}_{n=0}^{+\infty}$. For practical applications however, one shall use approximation $B_n(T,T+\tau)$ for a specific integer $n$. Fortunately, the first approximation ($n=0$) already appears to give very accurate results. This can be understood by taking into account that the first eigenvalue in the Karhunen-Lo\`{e}ve expansion accounts for the dominant part of the overall variance of the Ornstein-Uhlenbeck process. Based on this observation, we introduce a new series of zero-coupon bond price approximations.
\end{remark}

\begin{aprox}\label{ap8}
The sequence
\begin{equation}
\przybl{B_n}(T,T+\tau) \coloneqq \mathfrak{Q}_n(\exp(-I_{0,\tau})), \quad n\geq 1
\end{equation}
provides approximation
\begin{equation}
B_0(T,T+\tau)=\lim_{n\to +\infty} \przybl{B_n}(T,T+\tau).
\end{equation}
\end{aprox}

\begin{proof}
Note that $\przybl{I_{0,\tau}}(z)$ is a strictly positive, continuous function of $z \in \Real$. Therefore $\exp(-\przybl{I_{0,\tau}}(\cdotp)) \in \mathbb{L}^2\left(\Real, \frac{1}{\sqrt{2\pi}}e^{-\frac{1}{2}x^2}dx\right)$ and by \eqref{eq:233} we get our assertion.
\end{proof}\qed

Important property of such approximation is its monotonic dependence on $r_T$.

\begin{proposition}\label{prop11}
For any $n\geq 1$, $\przybl{B_n}(T,T+\tau)$ is a strictly decreasing function of $r_T$. Moreover
\begin{equation}
\lim_{r_T \to 0}\przybl{B_n}(T,T+\tau)=1,\qquad \lim_{\quad r_T \rightarrow +\infty}\przybl{B_n}(T,T+\tau)=0.
\end{equation}
\end{proposition}

\begin{proof}
By definition
\begin{equation}
\przybl{B_n}(T,T+\tau)= \mathfrak{Q}_n(\exp(-I_{0,\tau})) = \sum_{k=1}^{n} w_{n,k}\exp(-I_{0,\tau}(h_{n,k})).
\end{equation}
Since weights $w_{n,k}$ are positive, then it is sufficient to prove that $I_{0,\tau}(h_{n,k})$
is a strictly increasing function of $r_T$ for each $k$. Recalling the formulae
\begin{equation}
I_{0,\tau}(h_{n,k}) = \int_0^\tau \bar{r}_{T,t} G_{0,\tau}(t) F_{0,\tau}(t,h_{n,k})\,dt, \quad 1 \leq k \leq n,
\end{equation}
we note that $G_{0,\tau},F_{0,\tau}$ are strictly positive functions independent of $r_T$, whereas $\bar{r}_{T,t}$ is a strictly increasing function of $r_T$. Consequently, $I_{0,\tau}(h_{n,k})$ is a strictly increasing function of $r_T$.

Moreover, since $G_{0,k\tau}(\cdot), F_{0,k\tau}(\cdot,h_{n,j})$ are bounded on $[T, T+\tau]$ and convergences
\begin{equation}
\lim_{r_T \rightarrow 0}\bar{r}_{T,t} = 0, \qquad \lim_{r_T \to +\infty}\bar{r}_{T,t} = +\infty
\end{equation}
are uniform in $t \in [T,T+\tau]$, we obtain
\begin{equation}
\lim_{r_T \to 0} \exp(-I_{0,N\tau}(h_{k,j}))=1, \quad \lim_{r_T \to +\infty} \exp\big(-I_{0,N\tau}(h_{k,j})\big)=0.
\end{equation}
Since by definition
\begin{equation}
\przybl{B_n}(T,T+k\delta)=\sum_{j=1}^{n} w_{n,j}\exp\big(-I_{0,k\tau}(h_{k,j})\big)
\end{equation}
and that $w_{n,j}$ sum up to $1$, we get our assertion.
\end{proof}\qed

\section{Approximation for swaptions}
Let us consider the swaption with expiry $T>0$, strike $S>0$ and underlying swap of tenor $\tau = N\delta$, where $n \ge 1$ is the number of fixed-leg payments and $\delta > 0$ is the length of payment period. Consider the parameter $\omega$ related to the swaption type, equal $\omega = 1$ for a payer swaption and $\omega = -1$ for a receiver swaption. 
\\
\begin{notation}
Denote the stochastic discount factor and its conditional expectation in respect to $X_T$ by
\begin{equation}
\beta(t) \coloneqq \exp\left(-\int_0^tr_s\,ds\right), \quad \warunk{\beta}(x) \coloneqq \E\big(\beta(T)|X_T=x\big).
\end{equation}

Denote the values of swap annuity, underlying swap rate and two auxiliary quantities as
\begin{equation}
A(t,T,N) \coloneqq \delta \sum_{k=1}^{N}B(t,T+k\delta),
\end{equation}
\begin{equation}\label{eq:235}
r(t,T,N) \coloneqq \frac{B(t,T)-B(t,T+N\delta)}{A(t,T,N)},
\end{equation}
\begin{equation}\label{eq:12}
C(T,N,S) \coloneqq B(T,T+N\delta) + S A(T,T,N),
\end{equation}
\begin{equation}
P(T,N,S;X_T) \coloneqq \warunk{\beta}(X_T)\Big(1-C(T,N,S;X_T)\Big),
\end{equation}
where the last argument $X_T$ expresses (implicit) dependence of the variable on $X_T$ and can be omitted when not relevant. 
\end{notation}

\begin{proposition}\label{prop2}
The swaption price $Swpt$ can be expressed by the formula
\begin{equation}\label{eq:11}
Swpt = \E\Big(\omega \mathds{1}_{\{\omega \geq \omega C(T,N,S)\}}P(T,N,S;X_T)\Big).
\end{equation}
\end{proposition}

\begin{proof}
Pricing the swaption under a spot measure as the expectation of its intrinsic value at expiry and using our notations \eqref{eq:235}-\eqref{eq:12}, we can express the theoretical price of a payer/receiver swaption as
\begin{equation}
Swpt=\E\Big(\beta(T)A(T,T,N)\Big(\omega\Big(r(T,T,N)-S\Big)\Big)^+\Big)
=\E\bigg(\beta(T)\Big(\omega\Big(1-C(T,N,S)\Big)\Big)^+\Big).\\
\end{equation}
By the tower property, this expectation can be calculated taking conditional expectation with respect to $X_T$ first. Since $C(T,N,S)$ is $X_t$-measurable this leads to our assertion.
\end{proof}\qed

To obtain the approximation for swaption price we will apply approximations to $\warunk{\beta}(X_T)$ and $C(T,N,S)$ in \eqref{eq:11}. Let us start with approximating $\warunk{\beta}(X_T)$. \\

\begin{notation}
Let $m \geq 1$. We define the following functions
\begin{equation}	
\warunk{F_{m,T}}(t,x,z_1,\ldots,z_m) \coloneqq \exp\bigg(\sigma \sum_{k=1}^{m}\sqrt{\warunk{\lambda}_{k}(T)}\,\warunk{f}_{k,T}(t)\,z_k +\sigma\frac{K(t,T)}{V(T)}x\bigg),
\end{equation}

\begin{equation}
\warunk{G_{m,\tau}}(t) \coloneqq \exp\left(\frac{\sigma^2}{2}\Big(\warunk{V}(t)- \sum_{k=1}^{m}\warunk{\lambda}_{k}(T)\warunk{f}_{k,T}(t)^2\Big) \right),
\end{equation}

\begin{equation}
\warunk{I_{m,T}}(x,z_1,\ldots,z_m) \coloneqq \int_0^\tau \bar{r}_{T,t} \warunk{G_{m,\tau}}(t) \warunk{F_{m,\tau}}(t,x,z_1,\ldots,z_m))\,dt.
\end{equation}

\end{notation}

\begin{aprox}\label{ap1}
The sequence
\begin{equation}
\warunk{\beta_m}(x) \coloneqq \frac{1}{(\sqrt{2\pi})^{m}}\int_{-\infty}^{+\infty} \ldots \int_{-\infty}^{+\infty} \exp \left( -\warunk{I_{m,T}}(x,z_1,\ldots,z_m) -\frac{z_1^2 + \ldots + z_m^2}{2} \right) \,dz_1 \ldots dz_m
\end{equation}
provides approximation
\begin{equation}
\warunk{\beta}(x) = \lim_{m \rightarrow +\infty} \warunk{\beta_m}(x).
\end{equation}
\end{aprox}

\begin{proof}
The proof is analogous to that of Approximation \ref{ap7}, modified so that all random variables are replaced with their expectations conditional on $X_T$. All functions and constants are replaced with their corresponding "hat" versions, dependent on $x$. Such analogy is based on equation \eqref{eq:102} and the fact the Ornstein-Uhlenbeck bridge and all $\warunk{Z_k}$ are independent of $X_T$. 
\end{proof}\qed

\begin{aprox}\label{ap2}
The sequence
\begin{equation}
\przybl{\warunk{\beta_m}}(x) \coloneqq \mathfrak{Q}_m(\exp(-\warunk{I_{m,T}})(x,\cdot)), \quad m\geq 1
\end{equation}
approximates $\warunk{\beta_0}(x)$
\begin{equation}
\warunk{\beta_0}(x)=\lim_{m\to +\infty} \przybl{\warunk{\beta_m}}(x).
\end{equation}
\end{aprox}
\begin{proof}
The proof is a direct analogue of that presented in Approximation \ref{ap8}.
\end{proof}\qed

\begin{notation}
For $m,n \geq 1 $ let us denote
\begin{equation}
P_{m,n}(T,N,S) \equiv P_{m,n}(T,N,S;X_T) \coloneqq \przybl{\warunk{\beta_m}}(X_T)\Big(1-\przybl{C_n}(T,N,S;X_T)\Big),
\end{equation}
\begin{equation}
\przybl{C_n}(T,N,S;X_T) \coloneqq \przybl{B_n}(T,T+N\delta)+S\delta \sum_{k=1}^{N}\przybl{B_n}(T,T+k\delta),
\end{equation}
where the last argument $X_T$ expresses (implicit) dependence on $X_T$, which is clear from the fact that each $\przybl{B_n}(T,T+k\delta)$ is a function of $r_T$ and $r_T=\bar{r}_{0,T}\exp(\sigma X_T)$ ($\bar{r}_{0,T}$ is a constant).
\end{notation}

\begin{proposition}\label{prop3}
For any $m,n \geq 1$ there exist $x_0=x_0(m,n)$ such that $P_{m,n}(T,N,S;\cdot)$ is negative on $(-\infty, x_0)$ and positive on $(x_0, +\infty)$.
\end{proposition}

\begin{proof} 
Note that $r_T$ is an increasing function of $x_T$ and 
\begin{equation}
\lim_{x_T \to -\infty}r_T = 0, \qquad \lim_{x_T \to +\infty}r_T = +\infty.
\end{equation}
Therefore using Proposition \ref{prop11} we observe that $\przybl{C_n}(T,N,S;x)$ is a strictly decreasing function of $X_T$ and
\begin{equation}
\lim_{x \to +\infty} \przybl{C_n}(T,N,S;x)=0, \quad \lim_{x \to -\infty} \przybl{C_n}(T,N,S;x)=1 + N\delta S>1.
\end{equation}
Hence there exists finite $x_0$ such that
\begin{equation}
\przybl{C_{n}}(T,N,S;x)  \geq 1, \quad x \leq x_0 \quad \mathrm{and} \quad
\przybl{C_{n}}(T,N,S;x)  \leq 1, \quad x \geq x_0,
\end{equation}
which concerning $\przybl{\hat{\beta}_m}(X_T) > 0$, is equivalent to our assertion.
\end{proof}\qed

\begin{proposition}\label{prop4}
For any $m,n \geq 1 $ the price $\textrm{Swpt} $ of the swaption can be approximated with the formula
\begin{equation}
\textrm{Swpt} \approx \przybl{\textrm{Swpt}_{m,n}} \coloneqq \E\Big( \omega \mathds{1}_{\{\omega X_T \geq \omega x_0 \}} P_{m,n}(T,N,S; X_T) \Big).
\end{equation}
\end{proposition}

\begin{proof}
Approximations \ref{ap7} - \ref{ap2} directly imply that
\begin{equation}
\begin{split}
P(T,N,S; X_T) &\approx 
\warunk{\beta_0}(X_T)\big(1-B_0(T,T+N\delta)-S\delta \sum_{k=1}^{N}B_0(T,T+k\delta)\big) \\ &\approx
\przybl{\warunk{\beta_m}}(X_T)(1-\przybl{C_n}(T,N,S;X_T))=P_{m,n}(T,N,S; X_T).
\end{split}
\end{equation}
In addition, from Proposition \ref{prop3} we deduce that the condition $\omega P(T,n,S; X_T) \geq 0$ is equivalent to $\omega X_T \geq \omega x_0 $, hence finally
\begin{equation}
\textrm{Swpt}\approx \E\Big(\mathds{1}_{\{\omega X_T \geq \omega x_0 \}} P_{m,n}(T,N,S; X_T) \Big),
\end{equation}
which is our assertion
\end{proof}\qed

Now we are ready to provide a tractable formula for approximate swaption pricing. To this end we simply replace $\przybl{\textrm{Swpt}_{m,n}}$ in the last approximation with its specific Lagrange interpolating polynomials.

\begin{aprox}\label{ap5}
For $m,n,k \geq 1$ denote
\begin{equation}
f_{m,n}: f_{m,n}(z) = P_{m,n}(T,N,S; \sqrt{V(T)} z),
\end{equation}
\begin{equation}
\przybl{f_{m,n,k}} \coloneqq \mathfrak{L}_{k}(f_{m,n}).
\end{equation}
Take any $k,l \geq 1$ such that
\begin{equation}
f_{m,n}(h_{k,l}) < 0 \leq f_{m,n}(h_{k,l+1})
\end{equation}
and $\przybl{z_k} \in (h_{k,l},h_{k,l+1}]$ such that
\begin{equation}
\przybl{f_{m,n,k}}(\przybl{z_k})=0.
\end{equation}
Then the sequence
\begin{equation}
\przybl{Swpt_{m,n,k}} \coloneqq \frac{1}{\sqrt{2\pi}} \left\{ \left[ \mathfrak{H}(\przybl{f_{m,n,k}})\right](\przybl{z_k}) \exp \left(-\frac{1}{2}\przybl{z_k}^2 \right) + \omega \mathfrak{h} (\przybl{f_{m,n,k}}) \Phi(-\omega\przybl{z_k})\right\}
\end{equation}
provides approximation
\begin{equation}
\przybl{Swpt_{m,n}}=\lim_{k \rightarrow +\infty} \przybl{Swpt_{m,n,k}}.
\end{equation}
In particular, the swaption price $\textrm{Swpt}$ can be approximated as
\begin{equation}
Swpt \approx \przybl{Swpt_{m,n,k}}.
\end{equation}
\end{aprox}

\begin{proof}
First of all note that required $k,l,\przybl{z_k}$ exist. Indeed 
\begin{equation}
-h_{k,1} = h_{k,k} \to +\infty, \quad k \to +\infty.
\end{equation} 
Hence by Proposition \ref{prop3} 
\begin{equation}
f_{m,n}(h_{k,1}) < 0 < f_{m,n}(h_{k,k})
\end{equation}
for almost all $k$ and 
\begin{equation}
f_{m,n}(h_{k,l}) < 0 \leq f_{m,n}(h_{k,l+1})
\end{equation}
for a specific $l$. Refering to the definition of $\przybl{f_{m,n,k}}$, this also means
\begin{equation}
\przybl{f_{m,n}}(h_{k,l}) < 0 \leq \przybl{f_{m,n}}(h_{k,l+1}),
\end{equation}
hence $\przybl{f_{m,n,k}}$ has a root $\przybl{z_k}$ in $(h_{k,l},h_{k,l+1}]$.

By Proposition \ref{prop4} and Lemma \ref{lem55} we have respectively
\begin{equation}
\przybl{Swpt_{m,n}} = \E\left(\omega \mathds{1}_{\{\omega Z \geq \omega z_0 \}} f(Z)\right),
\end{equation}
\begin{equation}
\przybl{Swpt_{m,n,k}} =  \frac{1}{\sqrt{2\pi}} \int_{\przybl{z_k}}^{\omega\infty} \przybl{f_{m,n,k}}(z) e^{-\frac{1}{2}z^2} dz=\E\left(\omega\mathds{1}_{\{\omega Z\geq \omega\przybl{z_k}\}}\przybl{f_{m,n,k}}(Z)\right),
\end{equation}
where $Z \sim N(0,1)$, $z_0  \coloneqq \frac{x_0(m,n)}{\sqrt{V(T)}}.$ 
Then denoting
\begin{equation}
\epsilon_k \coloneqq \E\left(\omega \left(\mathds{1}_{\{\omega Z \geq \omega z_0 \}}-\mathds{1}_{\{\omega Z\geq \omega \przybl{z_k}\}}\right) f_{m,n}(Z)\right),
\end{equation}
\begin{equation}
\delta_k \coloneqq \E\left(\omega \mathds{1}_{\{\omega Z\geq \omega\przybl{z_k}\}}(f_{m,n}(Z)-\przybl{f_{m,n,k}}(Z))\right),
\end{equation}
we have
\begin{equation}\label{eq:333}
\przybl{Swpt_{m,n}} - \przybl{Swpt_{m,n,k}} = \epsilon_k + \delta_k.
\end{equation}
We will show that $\epsilon_k, \delta_k \to 0$. First note that
\begin{equation}
|\mathbb{P}(\omega Z \geq \omega z_0) - \mathbb{P}(\omega Z\geq \omega\przybl{z_k})| = |\Phi(z_0) - \Phi(\przybl{z_k})| \leq \frac{1}{\sqrt{2\pi}} |z_0 - \przybl{z_k}|.
\end{equation}
Additionally, from elaborations in Proposition \ref{prop3}
\begin{equation}
|P_{m,n}(T,N,S;X_T)| < \max\{1,N\delta S\} < 1 + N\delta S,
\end{equation}
hence $f_{m,n} \leq 1 + N\delta S$. Combining those together
\begin{equation}
|\epsilon_k| \leq \frac{1}{\sqrt{2\pi}} (1+N\delta S)|z_0 - \przybl{z_k}|.
\end{equation}
By definition of $x_0$ from Proposition \ref{prop3}, we note that $z_0 \in (h_{k,l},h_{k,l+1}]$. Taking in account that $\przybl{z_k}$ also lays in this interval, we have
\begin{equation}
|\epsilon_k| \leq \frac{1}{\sqrt{2\pi}} (1+N\delta S)(h_{k,l+1}-h_{k,l})
\end{equation}
and noting that $h_{k,l+1}-h_{k,l} \rightarrow 0$ as $k \to +\infty$
(see Theorem 6.1.2 \cite{szego}) we get $\epsilon_k \rightarrow 0$.\\
Now, let us observe that
\begin{equation}
\delta^2 \leq \E\left|f(Z)-\przybl{f_{m,n,k}}(Z)\right| \leq \lVert f-\przybl{f_{m,n,k}}\rVert,
\end{equation}
where $\lVert \cdot \rVert$ denotes the norm in $\mathbb{L}^2\left(\Real,\frac{1}{\sqrt{2\pi}}e^{-\frac{1}{2}x^2}dx\right)$.
As $f_{m,n}$ is a continuous and bounded function, thus by Lemma \ref{lem666} it holds $\lVert\przybl{f_{m,n,k}} - f_{m,n} \rVert \to 0$ as $k \to +\infty$, hence $\delta_k \to 0$.
Finally, since $\epsilon_k + \delta_k \rightarrow 0$, then taking into account Proposition \ref{prop4} and \eqref{eq:333}, we finally get our assertion.
\end{proof}\qed

\begin{remark}
Despite some complexity in the approximation formulas obtained, they are easily computable in practice. Namely, proper $l$ can be found by evaluating \\ $\przybl{P_{m,n}}(T,N,S;\sqrt{V(T)} h_j)$ at zeros of Hermite polynomial ($j=1,\ldots,k$), which can be easily done numerically. Those same values are also used for calculating the coefficients of polynomial $\przybl{f_{m,n,k}}$, specifically by solving a set of linear equations. Afterwards, coefficients of the polynomial $\mathfrak{H}(\przybl{f_{m,n,k}})$ and the value $\mathfrak{h}(\przybl{f_{m,n,k}})$ are directly computable by using recursive formulas given by their definition. Finally, the root $\przybl{z_k}$ of $\przybl{f_{m,n,k}}$ can be effectively found with any standard numerical procedure (such as Newton's method or the false position method). The only numerically extensive element here is the evaluation of $\przybl{P_{m,n}}$, which requires $k(Nn + m)$ numerical integrations ($k$ arguments $h_j$, $N$ maturities of bonds approximations, each consisting of $n$ components of quadrature, plus $m$ nodes of quadrature for approximation of $\warunk{\beta}$). 
\end{remark}

\begin{remark}
For practical applications, one should take some specific $k, m, n$ values. In typical situations we recommend using  $k = m = n = 5$, which proved to be accurate enough in our numerical tests and require a moderate number of numerical integrations.
\end{remark}

\begin{remark}\label{rem:13}
Note that for exact quantities it holds
\begin{equation}
\E\Big( \warunk{\beta}(X_T) \Big(1-C(T,n,S)\Big)\Big) = 1-C(0,n,S).
\end{equation}
However, it does not remain true if we substitute $\warunk{\beta}(X_T)$ and $C(T,n,S)$ with their approximations applied in Approximation \ref{ap5}. Consequently, our approximations do not obey the put-call parity exactly.
\end{remark}

\section{Numerical results}

In order to test the accuracy of the approximations presented in previous sections, several prices of zero coupon bonds have been computed. For simplicity $a(\cdot)$ was assumed to be a constant function of time of the form $a(t)=b\ln(r_{\mathrm{avg}})$, where $r_{\mathrm{avg}}=3\%$. For better clarity and comparability, results are presented in the form of yields-to-maturity (with a continuous compounding convention), not actual prices of bonds. In order to examine the dependence of results on different parameters, yields-to-maturity were calculated for the following:
\begin{itemize}
\item[$\bullet$]maturities: 1, 2, 5, 10 and 20 years
\item[$\bullet$]values of $r_0$: 1\%, 3\% and 6\%
\item[$\bullet$]values of $b$: 0.02 and 0.1
\item[$\bullet$]values of $\sigma$: 25\% and 50\%
\end{itemize}

In Table 1 we present results of such calculations, obtained from Approximation \ref{ap8}, benchmarked to the exact results obtained via Monte-Carlo simulations. One can see, the errors of approximations are very small, in the order of at most a few basis points for every set of parameters examined.

Similar numerical tests were performed for swaption prices. Several prices of payer and receiver swaptions have been computed, for various sets of parameters of the model and the swaption payoff. We used the same sets of model parameters as for bonds. For each of them, we examined the accuracy of approximations for various swaptions, using parameters as follows:
\begin{itemize}
\item[$\bullet$]swaption expiries: 1, 2, 5 and 10 years,
\item[$\bullet$]underlying tenors: 1, 2, 5 and 10 years,
\item[$\bullet$]swaption moneyness: 10\%, 80\%, 90\%, 100\%, 110\%, 125\%, 150\%.
\end{itemize}
(moneyness being a quotient of the swaption Strike and ATM strike as the forward swap rate). Due to the dimensionality of the parameter space, we present results in 2 layers:
\begin{itemize}
\item [$\bullet$] dependence on moneyness for a given tenor (Table 2),
\item [$\bullet$] dependence on underlying tenor for ATM strike (Table 3).
\end{itemize}
In the case of strikes other than ATM, only out-of-the-money swaptions were concerned, i.e. the price of payer or receiver swaption was calculated depending on whether its strike was above or below the forward swap rate (such restriction can be imposed without loss of generality, because of the put-call parity).

As in the case of bonds, prices calculated using Approximation \ref{ap5} were compared to "exact" prices (calculated on the lattice). In addition to such checks, we examined the scale of put-call disparity resulting from approximation, as mentioned above in Remark \ref{rem:13}. To this end, in Table 4 we compared prices of ATM-payer and ATM-receiver swaptions (which in principle should be equal) obtained from Approximation \ref{ap5}. To allow for easier comparison of results for different swaptions, and to stay compatible with market conventions, all swaption prices were translated into their implied volatilities. Hence figures in tables represent differences between implied volatilities corresponding to compared swaption prices (approximate vs exact in the case of Tables 2 and 3 or payer vs receiver in Table 4).

As you can see, with only a few exceptions, errors of approximations range from -50 bp to +50 bp, where the vast majority are less than 10 bp in terms of absolute value, which is far below a typical bid-offer spread. Not surprisingly, the biggest errors are observed in the case of long expiry/tenor and/or high volatility. Similar observations address the put-call disparity, which appears negligible except in cases of the longest expiries and/or tenors, which reflect weaker efficiency of approximations when applied to longer time horizons.

Finally, we compared our approximations for bonds with those obtained via methods proposed in \cite{tourrucoo},\cite{capriotti2}. We used Table 1 from \cite{capriotti2} herein, which contains approximations of zero-coupon bond prices obtained with both those methods as well as from Monte Carlo simulations, as benchmark values. Calculations have been conducted for specific sets of parameters, namely:
\begin{itemize}
\item[$\bullet$]maturities: 0.1, 0.5, 1, 2 and 3 years,
\item[$\bullet$]$r_{0}=6\%,b=\ln(0.04), \sigma=85\%$.
\end{itemize}

Table 5 includes these results complemented with prices obtained from our Approximation \ref{ap8}. However, in order to keep our convention, we converted bond prices to their yields and present the results in the form of a difference vs benchmark (MC). The results in Table 5 reveal moderately good performance of Approximation \ref{ap8} in comparison to other approaches. Most importantly, it maintains a rather stable error rate while increasing bond maturity, whereas results from Ref [16] exhibit very different behaviour, with errors increasing strongly with maturity.
\newpage

\small
\begin{table}[H]
\caption{Yield-to-maturities obtained by Monte Carlo simulations (MC)
and Approximation \ref{ap8} (A2).}
\begin{center}
\begin{tabular}{lllr<{\kern1.7em}rr>$r<$}
\toprule
\multicolumn3c{Model Parameters}&\multicolumn1c{\multirow2*{Maturity}}&\multicolumn3c{Yield}\\\cmidrule(r){1-3}\cmidrule(l){5-7}
\multicolumn1c{$r_0$}&\multicolumn1c{$b$}&\multicolumn1c{$\sigma$}
&&\multicolumn1c{MC}&\multicolumn1c{A2}&\multicolumn1c{Error}\\
\midrule
&	&	&	1	&	1.071	\%&	1.071	\%&	0.000	\%\\
&	&	&	2	&	1.142	\%&	1.142	\%&	0.000	\%\\
1\%	&0.1	&25\%	&5	&	1.350	\%&	1.350	\%&	0.000	\%\\
&	&	&	10	&	1.663	\%&	1.663	\%&	0.000	\%\\
&	&	&	20	&	0.832	\%&	0.832	\%&	0.000	\%\\\midrule
&	&	&	1	&	3.043	\%&	3.043	\%&	0.000	\%\\
&	&	&	2	&	3.080	\%&	3.080	\%&	0.000	\%\\
3\%	&0.1	&25\%	&5	&	3.159	\%&	3.159	\%&	0.000	\%\\
&	&	&	10	&	3.221	\%&	3.222	\%&	0.001	\%\\
&	&	&	20	&	1.610	\%&	1.611	\%&	0.001	\%\\\midrule
&	&	&	1	&	5.885	\%&	5.885	\%&	0.000	\%\\
&	&	&	2	&	5.769	\%&	5.769	\%&	0.000	\%\\
6\%	&0.1	&25\%	&5	&	5.435	\%&	5.436	\%&	0.001	\%\\
&	&	&	10	&	4.971	\%&	4.975	\%&	0.004	\%\\
&	&	&	20	&	2.485	\%&	2.487	\%&	0.002	\%\\\midrule
&	&	&	1	&	1.027	\%&	1.027	\%&	0.000	\%\\
&	&	&	2	&	1.053	\%&	1.053	\%&	0.000	\%\\
1\%	&0.02	&25\%	&5	&	1.134	\%&	1.134	\%&	0.000	\%\\
&	&	&	10	&	1.264	\%&	1.264	\%&	0.000	\%\\
&	&	&	20	&	0.632	\%&	0.632	\%&	0.000	\%\\\midrule
&	&	&	1	&	3.046	\%&	3.046	\%&	0.000	\%\\
&	&	&	2	&	3.089	\%&	3.089	\%&	0.000	\%\\
3\%	&0.02	&25\%	&5	&	3.203	\%&	3.203	\%&	0.000	\%\\
&	&	&	10	&	3.331	\%&	3.333	\%&	0.001	\%\\
&	&	&	20	&	1.666	\%&	1.666	\%&	0.001	\%\\\midrule
&	&	&	1	&	6.048	\%&	6.048	\%&	0.000	\%\\
&	&	&	2	&	6.086	\%&	6.086	\%&	0.000	\%\\
6\%	&0.02	&25\%	&5	&	6.145	\%&	6.146	\%&	0.001	\%\\
&	&	&	10	&	6.075	\%&	6.081	\%&	0.006	\%\\
&	&	&	20	&	3.038	\%&	3.041	\%&	0.003	\%\\\midrule
&	&	&	1	&	1.120	\%&	1.120	\%&	0.000	\%\\
&	&	&	2	&	1.243	\%&	1.243	\%&	0.000	\%\\
1\%	&0.1	&50\%	&5	&	1.607	\%&	1.607	\%&	0.000	\%\\
&	&	&	10	&	2.104	\%&	2.107	\%&	0.003	\%\\
&	&	&	20	&	1.052	\%&	1.053	\%&	0.001	\%\\\midrule
&	&	&	1	&	3.178	\%&	3.178	\%&	0.000	\%\\
&	&	&	2	&	3.336	\%&	3.336	\%&	0.000	\%\\
3\%	&0.1	&50\%	&5	&	3.668	\%&	3.670	\%&	0.002	\%\\
&	&	&	10	&	3.872	\%&	3.882	\%&	0.009	\%\\
&	&	&	20	&	1.936	\%&	1.941	\%&	0.005	\%\\\midrule
&	&	&	1	&	6.137	\%&	6.137	\%&	0.000	\%\\
&	&	&	2	&	6.215	\%&	6.216	\%&	0.001	\%\\
6\%	&0.1	&50\%	&5	&	6.174	\%&	6.181	\%&	0.006	\%\\
&	&	&	10	&	5.747	\%&	5.774	\%&	0.026	\%\\
&	&	&	20	&	2.874	\%&	2.887	\%&	0.013	\%\\
\bottomrule
\end{tabular}
\end{center}
\end{table}

\small
\begin{table}[H]
\caption{Differences between implied volatilities corresponding to swaptions prices calculated using Approximation \ref{ap5} and pricing on lattice, calculated for various moneyness levels.}
\begin{center}
\newcommand\mc{\multicolumn1c}
\newcommand\mr{\multirow4*}
\begin{tabular}{lllr<{\kern.5em}r<{\kern.5em}r<{\kern.5em}*{8}{>$r<$}}
\toprule
\multicolumn3c{\multirow2*{Model Parameters}}&\multicolumn3c{\multirow2*{Swaption}}&
\multicolumn8c{Implied volatility error vs moneyness (\%ATMF)}\\
\cmidrule(l){7-14}
&&&&&&
\multicolumn4c{RECEIVER} & \multicolumn4c{PAYER}\\
\cmidrule(r){1-3}\cmidrule(lr){4-6}\cmidrule(lr){7-10}\cmidrule(l){11-14}
\mc{$r_0$}&\mc{$b$}&\mc{$\sigma$}
&\mc{Mat.}&\mc{Ten.}&\mc{Fwd IRS}&\mc{70\%}&\mc{80\%}&\mc{90\%}&\mc{100\%}&\mc{100\%}&\mc{110\%}&\mc{125\%}&\mc{150\%}\\\midrule
\mr{1.0\%}&\mr{0.10}&\mr{25\%}&1&1&1.22\%&0.00\%&0.00\%&0.00\%&0.00\%&0.00\%&0.00\%&0.00\%&0.00\%\\
&&&2&2&1.43\%&0.00\%&0.00\%&0.00\%&0.00\%&0.00\%&0.00\%&0.00\%&0.00\%\\
&&&5&5&1.99\%&0.00\%&0.00\%&0.01\%&0.01\%&0.01\%&0.01\%&0.00\%&0.00\%\\
&&&10&10&2.60\%&-0.01\%&0.00\%&0.01\%&0.01\%&0.02\%&0.02\%&0.01\%&0.00\%\\\midrule
\mr{3.0\%}&\mr{0.10}&\mr{25\%}&1&1&3.17\%&0.00\%&0.00\%&0.00\%&0.00\%&0.00\%&0.00\%&0.00\%&0.00\%\\
&&&2&2&3.25\%&0.00\%&0.00\%&0.00\%&0.00\%&0.00\%&0.00\%&0.00\%&0.00\%\\
&&&5&5&3.34\%&-0.01\%&0.00\%&0.01\%&0.01\%&0.02\%&0.02\%&0.01\%&-0.01\%\\
&&&10&10&3.31\%&-0.01\%&0.00\%&0.00\%&0.00\%&0.03\%&0.03\%&0.02\%&0.01\%\\\midrule
\mr{6.0\%}&\mr{0.10}&\mr{25\%}&1&1&5.82\%&0.00\%&0.00\%&0.00\%&0.00\%&0.00\%&0.00\%&0.00\%&0.00\%\\
&&&2&2&5.46\%&-0.01\%&0.00\%&0.00\%&0.00\%&0.01\%&0.01\%&0.00\%&0.00\%\\
&&&5&5&4.62\%&-0.01\%&0.00\%&0.01\%&0.01\%&0.03\%&0.02\%&0.01\%&0.00\%\\
&&&10&10&3.86\%&-0.01\%&-0.01\%&-0.01\%&-0.01\%&0.03\%&0.03\%&0.03\%&0.02\%\\\midrule
\mr{1.0\%}&\mr{0.02}&\mr{25\%}&1&1&1.09\%&0.00\%&0.00\%&-0.01\%&-0.01\%&0.00\%&0.00\%&0.00\%&0.00\%\\
&&&2&2&1.17\%&0.00\%&0.00\%&0.00\%&0.00\%&0.00\%&0.00\%&0.00\%&0.00\%\\
&&&5&5&1.40\%&0.01\%&0.02\%&0.02\%&0.01\%&0.02\%&0.01\%&0.00\%&-0.01\%\\
&&&10&10&1.72\%&0.11\%&0.12\%&0.11\%&0.10\%&0.12\%&0.09\%&0.05\%&-0.01\%\\\midrule
\mr{3.0\%}&\mr{0.02}&\mr{25\%}&1&1&3.18\%&0.00\%&0.00\%&0.00\%&0.00\%&0.00\%&0.00\%&0.00\%&0.00\%\\
&&&2&2&3.30\%&0.00\%&0.00\%&0.00\%&0.00\%&0.00\%&0.00\%&0.00\%&0.00\%\\
&&&5&5&3.52\%&0.03\%&0.05\%&0.05\%&0.04\%&0.06\%&0.05\%&0.02\%&-0.02\%\\
&&&10&10&3.51\%&0.01\%&0.02\%&0.02\%&0.02\%&0.07\%&0.06\%&0.06\%&0.05\%\\\midrule
\mr{6.0\%}&\mr{0.02}&\mr{25\%}&1&1&6.32\%&0.00\%&0.00\%&0.00\%&0.00\%&0.00\%&0.00\%&0.00\%&0.00\%\\
&&&2&2&6.38\%&0.00\%&0.00\%&0.01\%&0.01\%&0.01\%&0.01\%&0.00\%&-0.01\%\\
&&&5&5&6.20\%&0.01\%&0.03\%&0.03\%&0.03\%&0.06\%&0.05\%&0.04\%&0.01\%\\
&&&10&10&5.34\%&-0.07\%&-0.13\%&-0.17\%&-0.18\%&-0.13\%&-0.12\%&-0.10\%&-0.02\%\\\midrule
\mr{1.0\%}&\mr{0.10}&\mr{50\%}&1&1&1.38\%&-0.02\%&-0.03\%&-0.04\%&-0.03\%&-0.02\%&-0.01\%&0.00\%&0.02\%\\
&&&2&2&1.75\%&-0.02\%&-0.02\%&-0.02\%&-0.01\%&-0.01\%&-0.01\%&0.00\%&0.01\%\\
&&&5&5&2.63\%&0.23\%&0.23\%&0.21\%&0.17\%&0.23\%&0.18\%&0.10\%&-0.02\%\\
&&&10&10&3.26\%&0.04\%&0.08\%&0.09\%&0.07\%&0.27\%&0.25\%&0.21\%&0.14\%\\\midrule
\mr{3.0\%}&\mr{0.10}&\mr{50\%}&1&1&3.55\%&0.00\%&-0.01\%&-0.01\%&-0.01\%&0.00\%&0.00\%&0.00\%&0.01\%\\
&&&2&2&3.90\%&0.07\%&0.08\%&0.07\%&0.06\%&0.08\%&0.06\%&0.02\%&-0.03\%\\
&&&5&5&4.16\%&0.23\%&0.26\%&0.25\%&0.22\%&0.32\%&0.27\%&0.20\%&0.07\%\\
&&&10&10&3.93\%&-0.06\%&-0.06\%&-0.07\%&-0.08\%&0.16\%&0.15\%&0.16\%&0.18\%\\\midrule
\mr{6.0\%}&\mr{0.10}&\mr{50\%}&1&1&6.49\%&0.01\%&0.02\%&0.02\%&0.01\%&0.03\%&0.02\%&0.01\%&-0.01\%\\
&&&2&2&6.43\%&0.12\%&0.15\%&0.14\%&0.12\%&0.15\%&0.12\%&0.06\%&-0.03\%\\
&&&5&5&5.49\%&0.10\%&0.13\%&0.13\%&0.12\%&0.27\%&0.25\%&0.21\%&0.15\%\\
&&&10&10&4.38\%&-0.10\%&-0.17\%&-0.22\%&-0.24\%&0.01\%&0.01\%&0.04\%&0.13\%\\\bottomrule
\end{tabular}
\end{center}
\end{table}

\small
\begin{table}[H]
\caption{Differences between implied volatilities corresponding to swaptions prices calculated using Approximation \ref{ap5} and pricing on lattice, calculated for various swaptions expiries and underlying tenors, for ATM strike.}
\begin{center}
\begin{tabular}{lllr<{\kern1em}*{4}{>$r<$}}
\toprule
\multicolumn3c{Model Parameters}
&\multicolumn1c{\multirow2*{\shortstack{Tenor\\Expiry}}}
&\multicolumn4c{ATMF Payer Volatility Error}\\
\cmidrule(r){1-3}\cmidrule(l){5-8}
\multicolumn1c{$r_0$}&\multicolumn1c{$b$}&\multicolumn1c{$\sigma$}
&&\multicolumn1c{1Y}&\multicolumn1c{2Y}&\multicolumn1c{5Y}&\multicolumn1c{10Y}\\
\midrule
&	&	&	1	&	0.00	\%&	0.00	\%&	0.00	\%&	-0.01	\%\\
\multirow2*{1\%}&	\multirow2*{0.1}&	\multirow2*{25\%}&	2	&	0.00	\%&	0.00	\%&	0.00	\%&	-0.01	\%\\
&	&	&	5	&	0.00	\%&	0.00	\%&	0.01	\%&	0.00	\%\\
&	&	&	10	&	0.01	\%&	0.02	\%&	0.01	\%&	0.01	\%\\\midrule
&	&	&	1	&	0.00	\%&	0.00	\%&	-0.01	\%&	-0.03	\%\\
\multirow2*{3\%}&	\multirow2*{0.1}&	\multirow2*{25\%}&	2	&	0.00	\%&	0.00	\%&	0.00	\%&	-0.02	\%\\
&	&	&	5	&	0.01	\%&	0.01	\%&	0.01	\%&	0.00	\%\\
&	&	&	10	&	0.02	\%&	0.02	\%&	0.01	\%&	0.00	\%\\\midrule
&	&	&	1	&	0.00	\%&	0.00	\%&	-0.02	\%&	-0.06	\%\\
\multirow2*{6\%}&	\multirow2*{0.1}&	\multirow2*{25\%}&	2	&	0.00	\%&	0.00	\%&	-0.01	\%&	-0.04	\%\\
&	&	&	5	&	0.02	\%&	0.02	\%&	0.01	\%&	-0.01	\%\\
&	&	&	10	&	0.02	\%&	0.02	\%&	0.01	\%&	-0.01	\%\\\midrule
&	&	&	1	&	-0.01	\%&	0.00	\%&	-0.01	\%&	-0.02	\%\\
\multirow2*{1\%}&	\multirow2*{0.02}&	\multirow2*{25\%}&	2	&	-0.01	\%&	0.00	\%&	-0.01	\%&	0.00	\%\\
&	&	&	5	&	-0.01	\%&	0.00	\%&	0.01	\%&	0.02	\%\\
&	&	&	10	&	0.03	\%&	0.05	\%&	0.09	\%&	0.10	\%\\\midrule
&	&	&	1	&	0.00	\%&	0.00	\%&	-0.01	\%&	-0.04	\%\\
\multirow2*{3\%}&	\multirow2*{0.02}&	\multirow2*{25\%}&	2	&	0.00	\%&	0.00	\%&	0.00	\%&	-0.02	\%\\
&	&	&	5	&	0.02	\%&	0.03	\%&	0.04	\%&	0.02	\%\\
&	&	&	10	&	0.14	\%&	0.15	\%&	0.12	\%&	0.02	\%\\\midrule
&	&	&	1	&	0.00	\%&	0.00	\%&	-0.02	\%&	-0.08	\%\\
\multirow2*{6\%}&	\multirow2*{0.02}&	\multirow2*{25\%}&	2	&	0.00	\%&	0.01	\%&	0.00	\%&	-0.05	\%\\
&	&	&	5	&	0.05	\%&	0.06	\%&	0.03	\%&	-0.04	\%\\
&	&	&	10	&	0.11	\%&	0.08	\%&	-0.05	\%&	-0.18	\%\\\midrule
&	&	&	1	&	-0.03	\%&	-0.02	\%&	-0.01	\%&	-0.13	\%\\
\multirow2*{1\%}&	\multirow2*{0.1}&	\multirow2*{50\%}&	2	&	-0.05	\%&	-0.01	\%&	0.01	\%&	-0.07	\%\\
&	&	&	5	&	0.06	\%&	0.14	\%&	0.17	\%&	0.04	\%\\
&	&	&	10	&	0.46	\%&	0.49	\%&	0.32	\%&	0.07	\%\\\midrule
&	&	&	1	&	-0.01	\%&	0.01	\%&	-0.02	\%&	-0.20	\%\\
\multirow2*{3\%}&	\multirow2*{0.1}&	\multirow2*{50\%}&	2	&	0.02	\%&	0.06	\%&	0.04	\%&	-0.10	\%\\
&	&	&	5	&	0.31	\%&	0.35	\%&	0.22	\%&	0.01	\%\\
&	&	&	10	&	0.58	\%&	0.49	\%&	0.15	\%&	-0.08	\%\\\midrule
&	&	&	1	&	0.01	\%&	0.02	\%&	-0.05	\%&	-0.33	\%\\
\multirow2*{6\%}&	\multirow2*{0.1}&	\multirow2*{50\%}&	2	&	0.09	\%&	0.12	\%&	0.03	\%&	-0.18	\%\\
&	&	&	5	&	0.45	\%&	0.40	\%&	0.12	\%&	-0.10	\%\\
&	&	&	10	&	0.46	\%&	0.29	\%&	-0.10	\%&	-0.24	\%\\
\bottomrule
\end{tabular}
\end{center}
\end{table}

\small
\begin{table}[H]
\caption{Call-put disparity for corresponding ATM payer and receiver swaptions evaluated using Approximation \ref{ap5}, expressed as the difference of corresponding implied volatilities of payer and receiver swaptions.}
\begin{center}
\begin{tabular}{lllr<{\kern1em}*{4}{>$r<$}}
\toprule
\multicolumn3c{Model Parameters}
&\multicolumn1c{\multirow2*{\shortstack{Tenor\\Expiry}}}
&\multicolumn4c{Payer-Receiver ATMF Volatility}\\
\cmidrule(r){1-3}\cmidrule(l){5-8}
\multicolumn1c{$r_0$}&\multicolumn1c{$b$}&\multicolumn1c{$\sigma$}
&&\multicolumn1c{1Y}&\multicolumn1c{2Y}&\multicolumn1c{5Y}&\multicolumn1c{10Y}\\
\midrule
&	&	&	1	&	0.00	\%&	0.00	\%&	0.01	\%&	0.02	\%\\
\multirow2*{1\%}&	\multirow2*{0.1}&	\multirow2*{25\%}&	2	&	0.00	\%&	0.00	\%&	0.01	\%&	0.02	\%\\
&	&	&	5	&	0.00	\%&	0.00	\%&	0.00	\%&	0.02	\%\\
&	&	&	10	&	0.00	\%&	0.00	\%&	0.01	\%&	0.02	\%\\\midrule
&	&	&	1	&	0.00	\%&	0.00	\%&	0.02	\%&	0.08	\%\\
\multirow2*{3\%}&	\multirow2*{0.1}&	\multirow2*{25\%}&	2	&	0.00	\%&	0.00	\%&	0.02	\%&	0.06	\%\\
&	&	&	5	&	0.00	\%&	0.00	\%&	0.01	\%&	0.04	\%\\
&	&	&	10	&	0.00	\%&	0.00	\%&	0.01	\%&	0.03	\%\\\midrule
&	&	&	1	&	0.00	\%&	0.01	\%&	0.05	\%&	0.18	\%\\
\multirow2*{6\%}&	\multirow2*{0.1}&	\multirow2*{25\%}&	2	&	0.00	\%&	0.01	\%&	0.04	\%&	0.12	\%\\
&	&	&	5	&	0.00	\%&	0.00	\%&	0.02	\%&	0.07	\%\\
&	&	&	10	&	0.00	\%&	0.00	\%&	0.01	\%&	0.04	\%\\\midrule
&	&	&	1	&	0.01	\%&	-0.01	\%&	0.01	\%&	0.04	\%\\
\multirow2*{1\%}&	\multirow2*{0.02}&	\multirow2*{25\%}&	2	&	0.00	\%&	0.00	\%&	0.01	\%&	0.03	\%\\
&	&	&	5	&	0.00	\%&	0.00	\%&	0.00	\%&	0.02	\%\\
&	&	&	10	&	0.00	\%&	0.00	\%&	0.00	\%&	0.02	\%\\\midrule
&	&	&	1	&	0.00	\%&	0.00	\%&	0.03	\%&	0.11	\%\\
\multirow2*{3\%}&	\multirow2*{0.02}&	\multirow2*{25\%}&	2	&	0.00	\%&	0.00	\%&	0.02	\%&	0.08	\%\\
&	&	&	5	&	0.00	\%&	0.00	\%&	0.02	\%&	0.06	\%\\
&	&	&	10	&	-0.01	\%&	-0.01	\%&	0.01	\%&	0.05	\%\\\midrule
&	&	&	1	&	0.00	\%&	0.01	\%&	0.06	\%&	0.22	\%\\
\multirow2*{6\%}&	\multirow2*{0.02}&	\multirow2*{25\%}&	2	&	0.00	\%&	0.01	\%&	0.04	\%&	0.16	\%\\
&	&	&	5	&	0.00	\%&	0.00	\%&	0.03	\%&	0.10	\%\\
&	&	&	10	&	-0.01	\%&	-0.01	\%&	0.01	\%&	0.06	\%\\\midrule
&	&	&	1	&	0.01	\%&	0.01	\%&	0.06	\%&	0.35	\%\\
\multirow2*{1\%}&	\multirow2*{0.1}&	\multirow2*{50\%}&	2	&	0.00	\%&	0.00	\%&	0.06	\%&	0.28	\%\\
&	&	&	5	&	-0.01	\%&	0.00	\%&	0.06	\%&	0.25	\%\\
&	&	&	10	&	-0.15	\%&	-0.10	\%&	0.01	\%&	0.19	\%\\\midrule
&	&	&	1	&	0.00	\%&	0.02	\%&	0.15	\%&	0.66	\%\\
\multirow2*{3\%}&	\multirow2*{0.1}&	\multirow2*{50\%}&	2	&	0.01	\%&	0.02	\%&	0.13	\%&	0.52	\%\\
&	&	&	5	&	-0.03	\%&	-0.01	\%&	0.10	\%&	0.38	\%\\
&	&	&	10	&	-0.16	\%&	-0.09	\%&	0.02	\%&	0.23	\%\\\midrule
&	&	&	1	&	0.01	\%&	0.05	\%&	0.29	\%&	1.10	\%\\
\multirow2*{6\%}&	\multirow2*{0.1}&	\multirow2*{50\%}&	2	&	0.01	\%&	0.04	\%&	0.23	\%&	0.81	\%\\
&	&	&	5	&	-0.04	\%&	0.00	\%&	0.15	\%&	0.49	\%\\
&	&	&	10	&	-0.14	\%&	-0.08	\%&	0.02	\%&	0.25	\%\\
\bottomrule
\end{tabular}
\end{center}
\end{table}

\small
\begin{table}[H]
\caption{Yield-to-maturities of bonds obtained from various approximations, expressed and an error vs Monte Carlo simulations (MC). The respective approximations are derived from small volatility expansion (Ref \cite{tourrucoo}), the exponent expansion truncated to the first EE(1), second EE(2) and third term EE(3) and our Approximation \ref{ap8}.}
\begin{center}
\begin{tabular}{lllr<{\kern1em}*{4}{>$r<$}}
\toprule
\multicolumn1c{Ref \cite{tourrucoo} vs MC}
&\multicolumn1c{EE(1) vs MC}
&\multicolumn1c{EE(2) vs MC}
&\multicolumn1c{EE(3) vs MC}
&\multicolumn1c{A2 vs MC}\\
\midrule
-0.10\%&0.00\%&0.00\%&0.00\%&-0.02\%\\
-0.23\%&0.02\%&0.02\%&0.02\%&-0.01\%\\
-0.46\%&0.01\%&0.00\%&0.00\%&-0.07\%\\
-0.90\%&0.06\%&0.03\%&0.00\%&-0.13\%\\
-1.24\%&0.17\%&0.10\%&0.00\%&-0.08\%\\
\bottomrule
\end{tabular}
\end{center}
\end{table}

\normalsize

\section*{Appendix}

\textbf{Proof of Lemma \ref{lem55}}
\begin{proof}
At first note that operator $\mathfrak{H}$ and functional $\mathfrak{h}$ are linear. Namely, let $\alpha$ be any real number and
$W_1(z)=\sum_{k=0}^{n}w_{1,k}z^k, W_2(z)=\sum_{k=0}^{m}w_{2,k}z^k$ polynomials. By Remark 5, without loss of generality we assume they have the same degree. Denote 
\begin{equation}
[\mathfrak{H}(W_1)](z)=\sum_{k=0}^{n-1}u_{1,k}z^k,\quad[\mathfrak{H}(W_2)](z)=\sum_{k=0}^{n-1}u_{2,k}z^k,
\end{equation}
\begin{equation}
[\mathfrak{H}(\alpha W_1)](z)=\sum_{k=0}^{n-1}g_{k}z^k,
\quad [\mathfrak{H}(W_1+W_2)](z)=\sum_{k=0}^{n-1}h_kz^k,
\end{equation}
then by definition
\begin{equation}
u_{1,n}=u_{1,n+1}=u_{2,n}=u_{2,n+1}=g_{n}=g_{n+1}=h_{n}=h_{n+1}=0,
\end{equation}
in specific
\begin{equation}
h_{n}= u_{1,n}+u_{2,n},\quad h_{n+1}= u_{1,n+1}+u_{2,n+1}, \quad g_{n}=\alpha u_{1,n},\quad g_{1,n+1}=\alpha u_{1,n+1}.
\end{equation}
In addition, for $k: -1 \leq k \leq n-1$
\begin{equation}
g_{k+2}=\alpha u_{1,k+2},\quad h_{1,k+2}=u_{1,k+2}+u_{2,k+2},\quad h_{k+2}=u_{1,k+1}+u_{2,k+1},
\end{equation}
thus
\begin{equation}
g_{k}=(k+2)g_{k+2}+\alpha w_{1,k+1}=\alpha((k+2)g_{1,k+2}+w_{1,k+1})=\alpha u_{1,k},
\end{equation}
\begin{equation}
\begin{split}
&h_{k}=(k+2)h_{k+2}+w_{1,k+1}+ w_{2,k+1}\\
&=((k+2)u_{1,k+2}+w_{1,k+1})+((k+2)u_{2,k+2}+w_{2,k+1})=u_{1,k}+u_{2,k},
\end{split}
\end{equation}
hence by the induction step it is clear that linearity properties are satisfied.\\
Now define operator $\przybl{\mathfrak{H}}: \Real[x] \rightarrow \Real[x]$ and functional $\przybl{\mathfrak{h}}: \Real[x] \rightarrow \Real$ such that for any polynomial $W(z)$
\begin{equation}
-\int W(z) e^{-\frac{1}{2}z^2}dz= [\przybl{\mathfrak{H}} (W)](z) e^{-\frac{1}{2}z^2}- \przybl{\mathfrak{h}}(W)\Phi(z)+C.
\end{equation}
Let us note that $\przybl{\mathfrak{H}}$ and $\przybl{\mathfrak{h}}$ are properly defined and linear, thus the lemma postulates
\begin{equation}
\przybl{\mathfrak{H}}= \mathfrak{H}, \quad \przybl{\mathfrak{h}}=\mathfrak{h},
\end{equation}
hence by linearity of $\mathfrak{H},\przybl{\mathfrak{H}},\mathfrak{h},\przybl{\mathfrak{h}}$ it is sufficient to prove
\begin{equation}
\przybl{\mathfrak{H}}(z^n)= \mathfrak{H}(z^n),\quad \przybl{\mathfrak{h}}(z^n)= \mathfrak{h}(z^n)
\end{equation}
for any non-negative integer $n$. To this end observe that the recursion formula for coefficients of $\mathfrak{H}(z^n)$ provides
\begin{equation}
u_k = \{(n-k)\ mod \ 2\}\frac{(n-1)!!}{k!!}, \quad 0\leq k < n, 
\end{equation}
\begin{equation}
u_{-1}=\{(n+1)\ mod \ 2\}(n-1)!!,
\end{equation}
whereas from integration by parts
\begin{equation}
\begin{split}
&-\int z^ne^{-\frac{1}{2}z^2}dz= z^{n-1}e^{-\frac{1}{2}z^2}- (n-1)\int z^{n-2}e^{-\frac{1}{2}z^2}dz\\&
= z^{n-1}e^{-\frac{1}{2}z^2}+(n-1)z^{n-3}e^{-\frac{1}{2}z^2} - (n-1)(n-3)\int z^{n-4}e^{-\frac{1}{2}z^2}dz=\ldots \\ & =\left(\sum_{k=0}^{n-1}\{(n-k)\ mod\ 2\}\frac{(n-1)!!}{k!!}\right) e^{-\frac{1}{2}z^2}-\{(n+1)\ mod\ 2\}(n-1)!!\Phi(z)+C,
\end{split}
\end{equation}
hence coefficients of $\przybl{\mathfrak{H}}(z^n)$ match respective values $u_k$ and
$\przybl{\mathfrak{h}}(z^n)=u^{-1}.$
\end{proof}\qed

\begin{theorem}[Nevai \cite{Nevai}]\label{nevai}
Let $f$ be a continuous function defined on the real line and $L_n(f,x)$ the Lagrange interpolation polynomial interpolating $f$ at zeros of the non-probabilistic Hermite polynomial of degree $n$. Assume that $f$ satisfies
\begin{equation}\label{eq:007}
\lim_{x\rightarrow +\infty}f(x)(1+|x|)e^{-\frac{1}{2}x^2} = 0.
\end{equation}
Then
\begin{equation}
\lim_{n\rightarrow +\infty} \int_{-\infty}^{+\infty}\left|\left|f(x)-L_n(f,x)\right|e^{-\frac{1}{2}x^2}\right|^pdx = 0
\end{equation}
holds for every $p>1$.
\end{theorem}

\textbf{Proof of Lemma \ref{lem666}}
\begin{proof}
Consider the function
\begin{equation}
\warunk{f}: \warunk{f}(x)=f\left(\sqrt{2}x\right).
\end{equation}
Since zeros of the probabilistic Hermite polynomial are $\sqrt{2}$ times the corresponding zeros of the non-probabilistic Hermite polynomial (of the same degree), then setting $z=\sqrt{2}\warunk{z}$ we have the following identity
\begin{equation}
L_n(\warunk{f},\warunk{z})=\left[\mathfrak{L}_n(f)\right](z),
\end{equation}
where $\mathfrak{L}_n$ is specified in Theorem \ref{nevai}. By changing the variable we have
\begin{equation}\label{eq:951}
\int_{-\infty}^{+\infty} \Big(f(z)- \left[\mathfrak{L}_n(f)\right](z)\Big)^2e^{-\frac{1}{2}z^2}dz=\sqrt{2}
\int_{-\infty}^{+\infty}\Big(\warunk{f}(\warunk{z}) - L_n(\warunk{f},\warunk{z})\Big)^2e^{-\warunk{z}^2}d\warunk{z}.
\end{equation}
Notice that since $f$ is a continuous and globally bounded function, then so is $\warunk{f}$. It's clear that any bounded function satisfies condition \eqref{eq:007}, hence setting $p=2$ Theorem \ref{nevai} postulates
\begin{equation}
\lim_{n \to +\infty} \int_{-\infty}^{+\infty}\Big(\warunk{f}(\warunk{z})- L_n(\warunk{f},\warunk{z})\Big)^2e^{-\warunk{z}^2}d\warunk{z} = 0,
\end{equation}
thus \eqref{eq:951} implies our assertion.
\end{proof}

\begin{lemma}\label{lem4}
Let $K_1, K_2:[0,T]^2 \to \Real$ be $C^1$ functions and also let $f:[0,T] \to \Real$ and $K:[0,T]^2 \to \Real$ be continuous functions, such that
\begin{equation}
K(s,t)=K_1(s,t),\quad s<t,
\end{equation}
\begin{equation}
K(s,t)=K_2(s,t),\quad s>t.
\end{equation}

Denote
\begin{itemize}
\item[$\bullet$]function $K\ast f: \big(K\ast f\big)(s) \coloneqq \int_0^TK(s,t)f(t)\,dt$
\item[$\bullet$]functions $h'_s, h''_{ss}$ as the first and second-order partial derivative of $h(s,t)$ with respect to variable $s$ for any twice differentiable function $h(s,t)$
\item[$\bullet$]single argument function $\frac{\partial}{\partial}h(s,s)$ as the first-order derivative of $h(s,s)$ any differentiable single argument function $h(s,s)$ 
\end{itemize}
Then we have the following:\\\\
a) $K\ast f$ is differentiable on $[0,T]$ and 
\begin{equation}\label{eq:60}
\big(K\ast f\big)'(s)=\big(K_2(s,s)-K_1(s,s)\big)f(s)+\big(K'_s\ast
f\big)(s).
\end{equation}
b) Moreover if $K_1$ and $K_2$ belong to $C^2\big((0,T)^2\big)$ and
\begin{equation}\label{eq:18}
K_1(t,t)=K_2(t,t)\qquad \textrm{for all}\quad t\in[0,T],
\end{equation}
then $K\ast f$ is twice differentiable and
\begin{equation}\label{eq:62}
\big(K\ast f\big)''(s)=\left(\frac{\partial}{\partial
s}K_2(s,s)-\frac{\partial}{\partial
s}K_1(s,s)\right)f(s)+\big(K''_{ss}\ast f\big)(s).
\end{equation}
\end{lemma}
\begin{proof}
We have
\begin{equation}
\big(K\ast f\big)(s)=\int_0^sK_2(s,t)f(t)\,dt+\int_s^TK_1(s,t)f(t)\,dt
\end{equation}
Hence by differentiating integrals we get
\begin{equation}
\frac{d}{ds}\int_0^sK_2(s,t)f(t)\,dt=K_2(s,s)f(s)+\int_0^sK'_s(s,t)f(t)\,dt
\end{equation}
and
\begin{equation}
\frac{d}{ds}\int_s^TK_1(s,t)f(t)\,dt=-K_1(s,s)f(s)+\int_s^TK'_s(s,t)f(t)\,dt.
\end{equation}
Summing up those equations we get the first part of our assertion. Assuming that $K_1(s,s)=K_2(s,s)$, equation \eqref{eq:60} obtains the form 
\begin{equation}
\big(K\ast f\big)'(s)=\big(K'_s\ast f\big)(s),
\end{equation}
which can be rewritten 
\begin{equation}
\frac{d}{ds}\int_0^sK_2(s,t)f(t)\,dt+\frac{d}{ds}\int_s^TK_1(s,t)f(t)\,dt=
\int_0^sK'_s(s,t)f(t)\,dt+\int_s^TK'_s(s,t)f(t)\,dt.
\end{equation}
Differentiating it once more and replacing in \eqref{eq:60} function $K'_s(s,t)$ with $K''_s(s,t)$ we finally get \eqref{eq:62}.
\end{proof}\qed

\begin{corollary}\label{cor6}
Let $K$ and $\warunk{K}$ be covariance functions from \eqref{eq:59} and \eqref{eq:91} respectively and $f$ be a continuous function on $[0,T]$. Then we have the following\\\\
a) Function $K\ast f$ is differentiable on $[0,T]$ and 
\begin{equation}
\big(K\ast f\big)'(s)=-b\big(K\ast f\big)(s).
\end{equation}
b) Functions $K\ast f,\warunk{K}\ast f$ are twice differentiable on $(0,T)$ and
\begin{equation}
\big(K\ast f\big)''(s)=\left(\frac{\partial}{\partial
s}K_2(s,s)-\frac{\partial}{\partial
s}K_1(s,s)\right)f(s)+\big(K''_{ss}\ast f\big)(s),
\end{equation}
\begin{equation}
\big(\warunk{K}\ast f\big)''(s)=\left(\frac{\partial}{\partial
s}K_2(s,s)-\frac{\partial}{\partial
s}K_1(s,s)\right)f(s)+\big(\warunk{K}''_{ss}\ast f\big)(s).
\end{equation}
\end{corollary}
\begin{proof}
Notice that functions $K,\warunk{K}$ are jointly continuous functions such that $K$ can be represented by functions $K_1$ for $s<t$
and $K_2$ for $s>t$, where $K_1,K_2$ are of the form
\begin{equation}
K_1(s,t)=\frac{1}{2b}\left(e^{-b(t-s)}-e^{-b(t+s)}\right),
\end{equation}
\begin{equation}
K_2(s,t)=\frac{1}{2b}\left(e^{-b(s-t)}-e^{-b(t+s)}\right)
\end{equation}
and function $\warunk{K}$ can be represented by functions $\warunk{K}_1$ for $s<t$ and $\warunk{K}_2$ for $s>t$, where $\warunk{K}_1,\warunk{K}_2$ are of the form
\begin{equation}
\warunk{K}_1(s,t)=K_1(s,t)-\frac{2bK(s,T)K(t,T)}{1-e^{-2bT}},
\end{equation}
\begin{equation}
\warunk{K}_2(s,t)=K_2(s,t)-\frac{2bK(s,T)K(t,T)}{1-e^{-2bT}}.
\end{equation}
Moreover 
\begin{equation}
K_1(t,t)=\frac{1}{2b}\left(1-e^{-2bt}\right)=K_2(t,t)
\end{equation}
and
\begin{equation}
\warunk{K}_1(t,t)=\frac{1}{2b}\left(1-e^{-2bt}\right)-\frac{2bK(t,T)^2}{1-e^{-2bT}}
=\warunk{K}_2(t,t),
\end{equation}
thus assumptions in Lemma \ref{lem4} are met for functions $K_1,K_2,K$ and $\warunk{K}_1,\warunk{K}_2,\warunk{K}$.
Partial derivatives of functions $K_1,K_2$ satisfy
\begin{equation}
\frac{\partial}{\partial s}K_j(s,t)=-bK_j(s,t),\quad j=1,2,
\end{equation}
hence from Lemma \ref{lem4} we have the first part of the assertion.\\
Now, computing first and second derivatives provides 
\begin{equation}
\frac{\partial}{\partial s}K_2(s,s)-\frac{\partial}{\partial s}K_1(s,s)=\big(-b-b\big)\cdot\frac{1}{2b}e^{-b(s-s)}=-1,
\end{equation}
\begin{equation}
\frac{\partial}{\partial s}\warunk{K}_2(s,s)-\frac{\partial}{\partial s}\warunk{K}_1(s,s)=\frac{\partial}{\partial s}K_2(s,s)-\frac{\partial}{\partial s}K_1(s,s)=-1
\end{equation}
and
\begin{equation}
\frac{\partial^2}{\partial s^2}K_j(s,t)=b^2K_j(s,t),\quad j=1,2,
\end{equation}
\begin{equation}
\frac{\partial^2}{\partial s^2}\warunk{K}_j(s,t)=b^2\warunk{K}_j(s,t),\quad j=1,2.
\end{equation}
Therefore applying Lemma \ref{lem4} we obtain 
\begin{equation}\label{eq:63a}
\big(K\ast f\big)''(s)=-f(s)+b^2\big(K\ast f\big)(s)
\end{equation}
and
\begin{equation}
\label{eq:63b} \big(\warunk{K}\ast
f\big)''(s)=-f(s)+b^2\big(\warunk{K}\ast f\big)(s),
\end{equation}
which is the second part of our assertion.
\end{proof}\qed

\textbf{Proof of Theorem \ref{thm1}}
\begin{proof}
Due to Theorem \ref{Karh} we know that the Ornstein-Uhlenbeck process $(X_t)_{t\in[0,\tau]}$ has an expansion 
of the form \eqref{eq:61}, where $Z_n$ are independent $N(0,1)$ random variables. Hence it is sufficient to obtain closed formulas for $\lambda_{n}(\tau)$ and $f_{n,\tau}$. Let $K$ be the covariance function of the process $(X_t)_{t\in[0,\tau]}$ as in \eqref{eq:59}. We
proceed with finding eigenfunctions of the operator $\mathcal{T}$ associated with kernel $K$. Recalling \eqref{eq:69} we know that functions $f_{n,\tau}$ satisfy the equation
\begin{equation}
\mathcal{T}f_{n,\tau}=\lambda_{n}(\tau) f_{n,\tau}
\end{equation}
for some respective $\lambda_{n}(\tau)>0$ and by Lemma \ref{lem4} we can alternatively write
\begin{equation}\label{eq:64}
K\ast f_{n,\tau}=\lambda_{n}(\tau) f_{n,\tau}.
\end{equation} 
Differentiating both sides twice and making use of Corollary \ref{cor6} we get
\begin{equation}
-f_{n,\tau}(s)+b^2\big(K\ast f_{n,\tau}\big)(s)=\lambda_{n}(\tau) f_{n,\tau}''(s).
\end{equation}
Therefore
\begin{equation}
-f_{n,\tau}(s)+b^2\lambda_{n}(\tau) f_{n,\tau}(s)=\lambda_{n}(\tau) f_{n,\tau}''(s)
\end{equation}
or equivalently
\begin{equation}\label{eq:65}
(\lambda_{n}(\tau) b^2-1)f_{n,\tau}(s)=\lambda_{n}(\tau) f_{n,\tau}''(s),
\end{equation}
hence the solution is a linear combination of functions
$\exp(\pm i\omega_n t)$, where $\omega_n$ is such that
\begin{equation}\label{eq:71}
\omega_n^2=\frac{1-\lambda_{n}(\tau) b^2}{\lambda_{n}(\tau)}.
\end{equation}
Moreover function $f_{n,\tau}$ satisfies two boundary conditions. Firstly, we notice that $K(0,t)=0$ for $t\in [0,\tau]$, therefore
\begin{equation}\label{eq:66}
\lambda_{n}(\tau) f_{n,\tau}(0)=\int_0^\tau K(0,t)f_{n,\tau}(t)\,dt=0.
\end{equation}
Secondly, applying Corollary \ref{cor6}, we observe that
\begin{equation}\label{eq:67}
\lambda_{n}(\tau) f_{n,\tau}'(\tau)=\big(K\ast f_{n,\tau}\big)'(\tau)=-b\big(K_s\ast f_{n,\tau}\big)(\tau)=-b\lambda_{n}(\tau)f_{n,\tau}(\tau).
\end{equation}
Since $\omega_n$ can be real or pure imaginary, depending on the sign of the
right-hand side of \eqref{eq:71}, let us consider those two cases in detail:\\\\
\vspace{1ex}\noindent $1^\circ \quad \omega_n\in i\Real,\omega_n\neq 0$\\
In this case $\omega_n=i\warunk{\omega}_n$ for some $\warunk{\omega}_n\in\Real\setminus \{0\}$ and $f_{n,\tau}$ is a linear combination of functions $\exp(\pm \warunk{\omega}_n t)$. Hence considering condition \eqref{eq:66}, we deduce that $f_{n,\tau}$ is of the form  
\begin{equation}\label{eq:87}
f_{n,\tau}(t)=c_n\big(\exp(-\warunk{\omega}_nt)-\exp(\warunk{\omega}_n t)\big),
\end{equation}
where $c_n\in\Real\setminus\{0\}$ (since $f_{n,\tau}$ is a real, not identically equal zero function). Thus applying condition \eqref{eq:67} to function $f_{n,\tau}$, we obtain from \eqref{eq:87} the equation
\begin{equation}
\warunk{\omega}_n\big(\exp(-\warunk{\omega}_n\tau)+\exp(\warunk{\omega}_n \tau)\big)=b\big(\exp(-\warunk{\omega}_n\tau)-\exp(\warunk{\omega}_n \tau)\big),
\end{equation}
which cannot be satisfied since the left and right-hand side of the equation have opposite signs (as $b>0$). Consequently, this case provides no solutions.\\\\
\vspace{2ex}\noindent $2^\circ \quad \omega_n\in \Real$\\
In this case $f_{n,\tau}$ is a linear combination of functions $\exp(\pm i\omega_n t)$, so 
condition \eqref{eq:66} implies that
\begin{equation}\label{eq:101}
f_{n,\tau}(t)=c'_n\big(\exp(i\omega_n t)-\exp(-i\omega_n t)\big)=2ic'_n\sin(\omega_n t)=c_n\sin(\omega_n t),
\end{equation}
where $\omega_n\neq 0$ and $c_n=-\frac{1}{2}ic'_n$ for some  $c'_n\in\Real\setminus\{0\}$ (since $f_{n,\tau}$ is a real, not identically equal zero function). Hence relying on the form of $f_{n,\tau}$ from \eqref{eq:101}, condition \eqref{eq:67} provides
\begin{equation}\label{eq:68}
\omega_n\cos(\omega_n \tau)=-b\sin(\omega_n \tau),
\end{equation}
which after elementary calculations can be rewritten as
\begin{equation}\label{eq:55}
\xi(\omega_n \tau)=-b\tau.
\end{equation}
Setting $\xi(x) \coloneqq x\cot x$. Equation \eqref{eq:68} implies that
$\sin(\omega_n \tau)\neq 0$, thus function $\xi$ is well-defined and continuous on each interval
$\big(n\pi,(n+1)\pi\big)$, $n\in\mathbb{N}\cup\{0\}$. Moreover it is
strictly decreasing on each such interval. Indeed
\begin{equation}
\xi'(t)=\cot x-\frac{x}{\sin^2 x}=\frac{\sin(2x)-2x}{2\sin^2x}<0.
\end{equation}
Therefore for each $n\in\mathbb{N}\cup\{0\}$ function
$\xi_n \coloneqq \xi\big|_{\big(n\pi,(n+1)\pi\big)}$ is invertible. One can
also easily see that
$\xi_n\Big(\big(n+\frac{1}{2}\big)\pi,(n+1)\pi\Big)=(-\infty,0)$,
hence equation \eqref{eq:55} has for each
$n\in\mathbb{N}\cup\{0\}$ exactly one solution of the form
\begin{equation}
\omega_n \coloneqq \frac{1}{\tau}\xi_n^{-1}(-b\tau).
\end{equation}
The value of $\lambda_{n}(\tau)$ corresponding to each $\omega_n$, implied from \eqref{eq:71}, is given by
\begin{equation}\label{eq:16}
\lambda_{n}(\tau) \coloneqq \frac{1}{b^2+\omega_n^2},
\end{equation}
thus this case provides solutions to \eqref{eq:65}, that satisfy conditions 
\eqref{eq:66}, \eqref{eq:67} and have the general form 
\begin{equation}\label{eq:15}
f_{n,\tau}(t)=c_n\sin(\omega_nt), \quad \lambda_n=\frac{1}{b^2+\omega_n^2} \quad n=1,2,\ldots
\end{equation}
Now, in order to prove that functions $f_{n,\tau}$ and corresponding $\lambda_{n}(\tau)$ are respectively the eigenfunctions and eigenvalues satisfying \eqref{eq:64}, for each $n\in\mathbb{N}\cup\{0\}$ we define function $g_n$ such that
\begin{equation}
\quad g_n(t)=\big(K\ast f_n\big)(t)-\lambda_{n}(\tau) f_{n,\tau}(t)
\end{equation}
and prove that it is identically equal zero. First of all, from the properties of $f_{n,\tau}$ and $\lambda_{n}(\tau)$, we can easily deduce that
\begin{equation}\label{eq:31}
\begin{split}
& g''_n(s)= \big(K\ast f_{n,\tau}\big)''(s)-\lambda_{n}(\tau) f''_{n,\tau}(s) \\ & = (-f_{n,\tau}(s)+b^2\big(K\ast f_{n,\tau}\big)(s))-(-f_{n,\tau}(s)+b^2\lambda_{n}(\tau) f_{n,\tau}(s))=b^2g_n(s)
\end{split}
\end{equation}
and because function $g_n$ is real, then the solution to \eqref{eq:31} is of the form
\begin{equation}
g_n(s)=d_1\exp(bt)+d_2\exp(-bt), \quad d_1,d_2 \in \Real.
\end{equation}
Moreover we have
\begin{equation}\label{eq:41}
g_n(0)=\big(K\ast f_{n,\tau}\big)(0)-\lambda_{n}(\tau) f_{n,\tau}(0)=0
\end{equation}
and
\begin{equation}\label{eq:42}
\begin{split}
g'_n(\tau)&=\big(K\ast f_{n,\tau}\big)'(\tau)-\lambda f_{n,\tau}^{'}(\tau)\\&= -b\big(K\ast f_{n,\tau}\big)(\tau)+b\lambda_{n}(\tau)f_{n,\tau}(\tau)=-bg_n(\tau).
\end{split}
\end{equation}
Equation \eqref{eq:41} implies that $d_2=-d_1$, hence from \eqref{eq:42} we have
\begin{equation}
d_1(\exp(b\tau)+\exp(-b\tau))=-d_1(\exp(b\tau)-\exp(-b\tau)),
\end{equation}
which holds if and only if $d_1=0$, because for non-zero $d_1$ the left and right-hand side of the equation  have opposite signs (since $b>0$). Thus $d_1=d_2=0$ and functions $g_n$ are identically equal zero, as desired.\\ 
Finally, to obtain $f_{n,\tau}$ normalized (as required in the Karhunen-Lo\`{e}ve
expansion) we find the appropriate $c_n$ by the direct calculation:
\begin{equation}
1=\norm{f_n}^2=\int_0^\tau c_n^2\sin^2(\omega_nt)\,dt=c_n^2\left(\frac{\tau}{2}-\frac{\sin(2\omega_n\tau)}{4\omega_n}\right).
\end{equation}
Taking into account that
\begin{equation}
\sin(2\alpha)=\frac{2\tan\alpha}{1+\tan^2\alpha}
\end{equation}
and that by \eqref{eq:68}
\begin{equation}
\tan(\omega_n\tau)=-\frac{\omega_n}{b},
\end{equation}
we have
\begin{equation}
\sin(2\omega_n\tau)=-\frac{2\omega_nb}{b^2+\omega_n^2},
\end{equation}
which leads to the equation
\begin{equation}
1=c_n^2\left(\frac{\tau}{2}+\frac{1}{2}\cdot\frac{b}{b^2+\omega_n^2}\right),
\end{equation}
hence
\begin{equation}\label{eq:14}
c_n = \left(\frac{1}{2}\big(\tau+b\lambda_n(T)\big)\right)^{-1/2}=\sqrt{\frac{2}{\tau+b\lambda_n(T)}}.
\end{equation}
Substituting into $\eqref{eq:15}$ values $c_n$ and $\lambda_{n}(\tau)$ from \eqref{eq:14} and \eqref{eq:16} respectively, we finally get the assertion.
\end{proof}\qed

\textbf{Proof of Theorem \ref{thm2}}

\begin{proof}
Due to Theorem \ref{Karh} we know that the Ornstein-Uhlenbeck bridge process $(\warunk{X}_t)_{t\in[0,T]}$ has an expansion of the form \eqref{eq:93}, where $\warunk{Z}_n$ are independent $N(0,1)$ random variables. Hence it is sufficient to obtain closed formulas for $\warunk{\lambda}_{n}(T)$ and $\warunk{f}_{n,T}$. Let $\warunk{K}$ be the covariance function of the process $(\warunk{X}_t)_{t\in[0,T]}$ as in \eqref{eq:91}. We
proceed with finding eigenfunctions of the operator $\mathcal{T}$ associated with kernel $\warunk{K}$. Recalling \eqref{eq:69} we know that functions $\warunk{f}_{n,T}$ satisfy the equation
\begin{equation}
\mathcal{T}\warunk{f}_{n,T}=\warunk{\lambda}_{n}(T) \warunk{f}_{n,T}
\end{equation}
for some respective $\warunk{\lambda}_{n}(T)>0$ and by Lemma \ref{lem4} we can alternatively write
\begin{equation}\label{eq:96}
\big(\warunk{K}\ast \warunk{f}_{n,T}\big)(s)=\warunk{\lambda}_{n}(T) \warunk{f}_{n,T}(s).
\end{equation}
Differentiating both sides twice and making use of Corollary \ref{cor6} we get
\begin{equation}
-\warunk{f}_{n,T}(s)+b^2\big(\warunk{K}\ast \warunk{f}_{n,T}\big)(s)=\warunk{\lambda}_{n}(T) \warunk{f}_{n,T}''(s).
\end{equation}
Therefore
\begin{equation}
-\warunk{f}_{n,T}(s)+b^2\warunk{\lambda}_{n}(T) \warunk{f}_{n,T}(s)=\warunk{\lambda}_{n}(T) \warunk{f}_{n,T}''(s)
\end{equation}
or equivalently
\begin{equation}\label{eq:92}
\left(\warunk{\lambda}_{n}(T) b^2-1\right)\warunk{f}_{n,T}(s)=\warunk{\lambda}_{n}(T) \warunk{f}_{n,T}''(s),
\end{equation}
hence the solution is a linear combination of functions
$\exp(\pm i\warunk{\omega}_n t)$, where $\warunk{\omega}_n$ is such that
\begin{equation}\label{eq:95}
\warunk{\omega}_n^2=\frac{1-\warunk{\lambda}_{n}(T) b^2}{\warunk{\lambda}_{n}(T)}.
\end{equation}
Moreover function $\warunk{f}_{n,\tau}$ satisfies two boundary conditions. Firstly, we notice that $\warunk{K}(0,t)=0$ for $t\in[0,T]$, therefore
\begin{equation}\label{eq:99}
\warunk{\lambda}_{n}(T) \warunk{f}_{n,T}(0)=\int_0^T \warunk{K}(0,t)\warunk{f}_{n,T}(t)\,dt=0.
\end{equation}
Secondly, we observe that $\warunk{K}(T,t)=0$ for $t\in[0,T]$, thus
\begin{equation}\label{eq:97}
\warunk{\lambda}_{n}(T) \warunk{f}_{n,T}(T)=\int_0^T\warunk{K}(T,t)\warunk{f}_{n,T}(t)\,dt=0.
\end{equation}
Since $\warunk{\omega}_n$ can be real or pure imaginary, depending on the sign of the right-hand side of \eqref{eq:95}, let us consider those two cases in detail\\\\
\vspace{1ex}\noindent $1^\circ \quad \warunk{\omega}_n\in i\Real,\warunk{\omega}_n\neq 0$\\
In this case $\warunk{\omega}_n=i\przybl{\omega_n}$ for some $\przybl{\omega_n}\in\Real\setminus \{0\}$ and $\warunk{f}_{n,\tau}$ is a linear combination of functions $\exp(\pm \przybl{\omega_n} t)$. Hence considering condition \eqref{eq:99}, we deduce that $\warunk{f}_{n,T}$ is of the form  
\begin{equation}\label{eq:100}
\warunk{f}_{n,T}(t)=c_n\big(\exp(-\przybl{\omega_n}t)-\exp(\przybl{\omega_n} t)\big),
\end{equation}
where $c_n\in\Real\setminus\{0\}$ (since $\warunk{f}_{n,T}$ is a real, not identically equal zero function). 
Thus applying condition \eqref{eq:97} to $\warunk{f}_{n,T}$, we obtain from \eqref{eq:100} the equation
\begin{equation}
c_n\big(\exp(-\przybl{\omega_n} T)-\exp(\przybl{\omega_n} T)\big)=0,
\end{equation}
which cannot be satisfied unless $c_n= 0$ or $\przybl{\omega_n}= 0$. Consequently, this case provides no solutions.\\\\
\vspace{2ex}\noindent $2^\circ \quad \warunk{\omega}_n\in \Real$\\
In this case $\warunk{f}_{n,T}$ is a linear combination of functions $\exp(\pm i\warunk{\omega}_n t)$, so 
condition \eqref{eq:99} implies that
\begin{equation}\label{eq:105}
\warunk{f}_{n,T}(t)=c'_n\big(\exp(i\warunk{\omega}_n t)-\exp(-i\warunk{\omega}_n t)\big)=2ic'_n\sin(\warunk{\omega}_n t)=c_n\sin(\warunk{\omega}_n t),
\end{equation}
where $c_n=-\frac{1}{2}ic'_n$ for some $c'_n\in\Real\setminus\{0\}$ (since $\warunk{f}_{n,T}$ is a real function). 
Hence, relying on the form of $\warunk{f}_{n,T}$ from \eqref{eq:99}, condition \eqref{eq:97} provides
\begin{equation}
c_n\sin(\warunk{\omega}_n T)=0
\end{equation}
has for each
$n\in\mathbb{N}$ exactly one solution of the form
\begin{equation}
\warunk{\omega}_n \coloneqq \frac{n\pi t}{T}.
\end{equation}
The value of $\warunk{\lambda}_{n}(T)$ corresponding to each $\warunk{\omega}_n$, implied from \eqref{eq:71}, is given by
\begin{equation}\label{eq:81}
\warunk{\lambda}_n(T) \coloneqq \frac{T^2}{b^2T^2+n^2\pi^2},
\end{equation}
thus this case provides solutions to \eqref{eq:92} that satisfy conditions \eqref{eq:99}, \eqref{eq:97} and have the general form
\begin{equation}\label{eq:36}
\warunk{f}_{n,T}(t)=c_n\sin\bigg(\frac{n\pi t}{T}\bigg), \quad \warunk{\lambda}_n(T)=\frac{T^2}{b^2T^2+n^2\pi^2}. 
\end{equation}
Now, in order to prove that functions $\warunk{f}_{n,T}$ and corresponding $\warunk{\lambda}_n(T)$ are respectively the eigenfunctions and eigenvalues satisfying \eqref{eq:96}, for each $n\in\mathbb{N}$ we define function $\warunk{g}_n$ such that
\begin{equation}
\warunk{g}_n(t)=\big(\warunk{K}\ast \warunk{f}_{n,T}\big)(t)-\warunk{\lambda}_n(T) \warunk{f}_{n,T}
\end{equation}
and prove that it is identically equal zero. From the properties of $\warunk{f}_{n,T}$ and $\warunk{\lambda}_n(\tau)$, we can easily deduce that
\begin{equation}\label{eq:37}
\begin{split}
& \warunk{g}''_n(s)= \big(\warunk{K}\ast \warunk{f}_{n,T}\big)''(s)-\warunk{\lambda}_n(T)\warunk{f}''_{n,T}(s) \\ & =
(-\warunk{f}_{n,T}(s)+b^2\big(\warunk{K}\ast \warunk{f}_{n,T}\big)(s))-
(-\warunk{f}_{n,T}(s)+b^2\warunk{\lambda}_n(T) \warunk{f}_{n,T}(s))=b^2\warunk{g}_n(s)
\end{split}
\end{equation}
and because function $\warunk{g}_n$ is a real function, then the solution to \eqref{eq:37} is of the form
\begin{equation}
\warunk{g}_n(s)=d_1\exp(bt)+d_2\exp(-bt), \quad d_1,d_2 \in \Real.
\end{equation}
In addition we have
\begin{equation}\label{eq:75}
\warunk{g}_n(0)=\big(\warunk{K}\ast \warunk{f}_{n,T}\big)(0)-\warunk{\lambda}_n(T) \warunk{f}_{n,T}(0)=0
\end{equation}
and
\begin{equation}\label{eq:72}
\warunk{g}_n(T)=\big(\warunk{K}\ast \warunk{f}_{n,T}\big)(T)-\warunk{\lambda}_n(T) \warunk{f}_{n,T}(T)=0.
\end{equation}
Equation \eqref{eq:75} implies that $d_2=-d_1$, hence from \eqref{eq:72} we have
\begin{equation}
d_1(\exp(bT)-\exp(-bT))=0,
\end{equation}
which holds if and only if $d_1=0$ (since $b>0$). Thus $d_1=d_2=0$ and functions $\warunk{g}_n$ are identically equal zero, as desired.\\ 
Finally, to obtain $f_{n,\tau}$ normalized (as required in the Karhunen-Lo\`{e}ve
expansion) we find the appropriate $c_n$ by the direct calculation
\begin{equation}
1=\norm{\warunk{f}_{n,T}}^2=\int_0^Tc_n^2\sin^2\bigg(\frac{n\pi t}{T}\bigg)\,dt=\frac{c_n^2T}{2},
\end{equation}
hence
\begin{equation}\label{eq:89}
c_n = \sqrt{\frac{2}{T}}.
\end{equation}
Substituting into \eqref{eq:36} values $c_n$ and $\warunk{\lambda}_n(T)$ from \eqref{eq:89} and \eqref{eq:81} respectively, we finally get the assertion.
\end{proof}\qed

\bibliographystyle{amsplain}
\bibliography{xbib}

\end{document}